\documentclass[11pt]{article} 

\usepackage{fullpage}
\usepackage{bbm}
\usepackage{graphicx}
\usepackage{upref}
\usepackage{enumerate}
\usepackage{latexsym}
\usepackage[colorlinks,linkcolor=red, citecolor=blue]{hyperref}
\usepackage{color,graphics}
\usepackage{comment} 

\usepackage{amsmath,amsthm,amssymb,algorithm,epsfig,graphicx}
\usepackage{algorithmicx}
\usepackage{algpseudocode}
\usepackage{mathtools}

\usepackage{amsmath}
\usepackage{amsfonts}
\usepackage{tikz} 

\usepackage{amssymb}
\usepackage{varioref}
\usepackage[ansinew]{inputenc}

\newtheorem{theorem}{Theorem}[section]
\newtheorem{lemma}[theorem]{Lemma}
\newtheorem{claim}[theorem]{Claim}

\newtheorem{corollary}[theorem]{Corollary}
\newtheorem{fact}[theorem]{Fact}

\newtheorem{defn}[theorem]{Definition}

\newtheorem{remark}[theorem]{Remark}

\long\def\rem#1{}
\def\B{\{0,1\}}
\newcommand{\Ker}{\text{Ker}}
\newcommand{\Img}{\text{Im}}
\newcommand{\Sym}{\text{Sym}}
\newcommand{\Span}{\text{Span}}

\newcommand{\Dim}{\text{Dim}}
\newcommand{\rank}{\text{rank}}

\newcommand{\R}{\mathbb{R}}

\newcommand{\psdrk}{\mathrm{rk_{psd}}}
\newcommand{\sosdeg}{\textrm{sos-deg}}
\DeclarePairedDelimiter{\ceil}{\lceil}{\rceil}
\DeclarePairedDelimiter{\floor}{\lfloor}{\rfloor}

\newcommand{\argmin}{\operatornamewithlimits{argmin}}

\newcommand{\eps}{\varepsilon}

\newcommand{\braket}[2]{\langle #1|#2\rangle}

\renewcommand{\part}[2]{\frac{\partial #1}{\partial #2}}

\newcommand{\all}[2]{\begin{align}\label{#2} #1\end{align}}
\newcommand{\al}[1]{\begin{align} #1\end{align}}

\newcommand{\en}[1]{\left ( #1 \right )}

\newcommand{\nl}{\notag \\}

\newcommand{\Tr}{\mathrm{Tr}}

\newcommand{\G}{\mathcal{G}}

\newcommand{\x}{\mathbf x}

\long\def\rem#1{}
\def\B{\{0,1\}}

\newcommand{\Ex}{\mathbb{E}}

\newcommand{\QE}{\mathrm{QE}}

\newcommand{\thmref}[1]{\hyperref[#1]{{Theorem~\ref*{#1}}}}
\newcommand{\lemref}[1]{\hyperref[#1]{{Lemma~\ref*{#1}}}}
\newcommand{\remref}[1]{\hyperref[#1]{{Remark~\ref*{#1}}}}
\newcommand{\corref}[1]{\hyperref[#1]{{Corollary~\ref*{#1}}}}
\newcommand{\eqnref}[1]{\hyperref[#1]{{Equation~(\ref*{#1})}}}
\newcommand{\claimref}[1]{\hyperref[#1]{{Claim~\ref*{#1}}}}
\newcommand{\remarkref}[1]{\hyperref[#1]{{Remark~\ref*{#1}}}}
\newcommand{\propref}[1]{\hyperref[#1]{{Proposition~\ref*{#1}}}}
\newcommand{\factref}[1]{\hyperref[#1]{{Fact~\ref*{#1}}}}
\newcommand{\defref}[1]{\hyperref[#1]{{Definition~\ref*{#1}}}}
\newcommand{\exampleref}[1]{\hyperref[#1]{{Example~\ref*{#1}}}}
\newcommand{\hypref}[1]{\hyperref[#1]{{Hypothesis~\ref*{#1}}}}
\newcommand{\secref}[1]{\hyperref[#1]{{Section~\ref*{#1}}}}
\newcommand{\chapref}[1]{\hyperref[#1]{{Chapter~\ref*{#1}}}}
\newcommand{\apref}[1]{\hyperref[#1]{{Appendix~\ref*{#1}}}}
\newcommand{\eqnsref}[1]{\hyperref[#1]{{Equations~(\ref*{#1})}}}
\newcommand{\algoref}[1]{\hyperref[#1]{{Algorithm~\ref*{#1}}}}

\AtBeginDocument{}

\begin{document}
\title{On the sum-of-squares degree of symmetric quadratic functions} 
\author{
Troy Lee\thanks{School of Physical and Mathematical Sciences, Nanyang Technological University and Centre for 
Quantum Technologies, Singapore.  This material is based on research supported by the Singapore National Research 
Foundation under NRF RF Award No. NRF-NRFF2013-13.  Email: {\tt troyjlee@gmail.com}, 
{\tt aprakash@ntu.edu.sg}} \and 
Anupam Prakash$^*$
\and 
Ronald de Wolf\thanks{QuSoft, CWI and University of Amsterdam, the Netherlands.
Partially supported by ERC Consolidator Grant QPROGRESS and by the European Commission FET-Proactive project Quantum Algorithms (QALGO) 600700. {\tt rdewolf@cwi.nl}}
\and 
Henry Yuen\thanks{MIT, Cambridge, MA. Supported by NSF grants 1218547 and 1452302 and a Simons Graduate Award in Theoretical Computer Science. {\tt hyuen@mit.edu}}
}
\maketitle
\thispagestyle{empty}

\begin{abstract}
We study how well functions over the boolean hypercube of the form $f_k(x)=(|x|-k)(|x|-k-1)$ can be approximated by sums of squares of low-degree polynomials, obtaining good bounds for the case of approximation in $\ell_{\infty}$-norm as well as in $\ell_1$-norm. We describe three complexity-theoretic applications:
(1) a proof that the recent breakthrough lower bound of Lee, Raghavendra, and Steurer~\cite{LRS15} on the positive semidefinite extension complexity of the correlation and TSP polytopes cannot be improved further by showing better sum-of-squares degree lower bounds on $\ell_1$-approximation of~$f_k$;
(2) a proof that Grigoriev's lower bound on the degree of Positivstellensatz refutations for the knapsack problem is optimal, answering an open question from~\cite{G01};
(3) bounds on the query complexity of quantum algorithms whose expected output approximates such functions.
\end{abstract}

\newpage

\setcounter{page}{1}

\section{Introduction} 

\subsection{Approximation of functions on the boolean hypercube by polynomials} 

Classical approximation theory studies how well a function $f:\R \rightarrow \R$ can be approximated by simpler functions, most commonly by polynomials of bounded degree.  Approximation theory has found applications throughout complexity theory, for example in learning theory~\cite{MP88,O'DS10}, query complexity~\cite{NS94,BBCMW01}, communication complexity~\cite{She09b,SZ09b}, and more.

An important special case is the investigation of the best approximation to a real boolean function $f: \{0,1\}^n \rightarrow \mathbb{R}$ in $\ell_\infty$-distance by a degree-$d$ polynomial in the $n$ variables $x_1,\ldots,x_n$.  Nisan and Szegedy~\cite{NS94} initiated this study, showing that any polynomial that approximates the OR function with constant error in $\ell_\infty$-norm on $\{0,1\}^n$ has degree $\Omega(\sqrt{n})$. They also showed this bound is tight by constructing an $O(\sqrt{n})$-degree approximating polynomial for the OR function from a Chebyshev polynomial.  Paturi~\cite{P92} followed this by characterizing the approximate degree of 
\emph{any} symmetric boolean function, i.e., any function $f:\{0,1\}^n \rightarrow \{0,1\}$ which only depends on the number of ones $|x|$ in the $n$-bit input~$x$, and not on their locations. 
To get a feel for Paturi's theorem, consider the special case 
of a function $g_k:\{0,1\}^n \rightarrow \{0,1\}$ where $g_k(x)=0$ unless $|x|=k$ in which case $g_k(x)=1$.  Paturi's theorem says that the 
$\tfrac{1}{4}$-error approximate degree of $g_k$, denoted by $\deg_{1/4}(g_k)$, is 
$\Theta(\sqrt{k(n-k)})$.  Later, the $\ell_\infty$-approximate degree of symmetric boolean functions was characterized for all approximation errors~$\eps$ by \cite{She09,deW08}.  Again in the special case of~$g_k$, these results say that the degree of a polynomial that approximates $g_k$ up to error $\eps\geq 2^{-n}$ is 
$\Theta(\sqrt{k(n-k)} + \sqrt{n \log(1/\eps)})$.

\subsection{Our results on sum-of-squares approximation}

Here we study the representation of non-negative functions on the boolean hypercube by \emph{sums of squares} of polynomials. More precisely, a non-negative boolean function $f: \{0,1\}^n \to \R_+$ has an (exact) \emph{degree-$d$ sum-of-squares (sos) representation} if there exist degree-$d$ polynomials $h_1,\ldots,h_r$ over the reals such that for all $x \in \{0,1\}^n$,
$$
	f(x) = h_1(x)^2 + \cdots + h_r(x)^2.
$$
Let $\sosdeg(f)$ be the minimum~$d$ such that a non-negative function $f$ has a degree-$d$ sum-of-squares representation.\footnote{Note that the degree of the polynomial representing $f$ will actually be $2d$.} This sos degree is an important quantity that arises in the context of optimization and proof complexity, as also witnessed by our applications below. 

The obvious fact that a sum of squares of polynomials is globally non-negative is remarkably useful.  For example, for a graph $G=([n], E)$, if $f(x_1, \ldots, x_n) = c - \sum_{(i,j) \in E} (x_i -x_j)^2$ has an sos representation on 
the boolean cube, then $c\geq \sum_{(i,j) \in E} (x_i -x_j)^2$ for all $x\in \{0,1\}^n$, and hence $G$ has no cut of size larger than~$c$.  Moreover if $f$ has a degree-$d$ sos representation for small~$d$, then this provides a small \emph{certificate} (of size $n^{O(d)}$) that $f$ has no cut of size larger than $c$.  
Such certificates can in fact be found by means of semidefinite programming; these observations are the basis of the semidefinite programming hierarchies of Lasserre and Parrilo \cite{Shor87, Lasserre01, Parrilo00} that have been the 
subject of intense study in approximation algorithms.   

While exact sum-of-squares degree of functions on the boolean hypercube has been previously studied, there has been little work on 
the \emph{approximation} of such functions by sos polynomials. This is the focus of our paper, and we prove a number of tight bounds on the \emph{approximate sum-of-squares degree} of functions on the hypercube. We consider two notions of approximation in this paper. The most familiar is $\ell_\infty$-approximation: an sos polynomial $h$ $\eps$-approximates a function $f : \{0,1\}^n \to \R_+$ in $\ell_\infty$-distance if $| f(x) - h(x) | \leq \eps$ for all $x \in \{0,1\}^n$. We let $\sosdeg_\eps(f,\ell_\infty)$ denote the minimum degree of an sos polynomial that $\eps$-approximates $f$ in $\ell_\infty$-distance. The other notion is $\ell_1$-approximation: an sos polynomial $h$ $\eps$-approximates $f$ in $\ell_1$-distance if $\sum_{x \in \{0,1\}^n} |f(x) - h(x)| \leq \eps$, and we let $\sosdeg_\eps(f,\ell_1)$ denote the minimum degree of an sos polynomial that $\eps$-approximates $f$ in $\ell_1$-distance. Note that $\eps=\delta 2^n$ corresponds to \emph{average} approximation error~$\delta$.\footnote{Also note that the existence of a degree-$d$ sos approximation in either of these notions can be formulated as the feasibility of a semidefinite program of size polynomial in the domain size $2^n$, as follows. For $x\in\{0,1\}^n$, let $m_x$ be the column vector of dimension $D=\sum_{i=0}^d{n\choose i}$ indexed by sets $S\subseteq[n]$ of size $\leq d$, with entry $m_{x,S}=\prod_{i\in S}x_i$. Let $M_x$ be the $D\times D$ rank-1 matrix $m_x m_x^T$. Suppose $p(x)=\sum_{S:|S|\leq d}p_S\prod_{i\in S}x_i$ is a multilinear polynomial of degree~$d$, where with slight abuse of notation we use $p$ also to denote the $D$-dimensional vector of real coefficients~$p_S$. Then $p(x)$ is the inner product of $p$ and $m_x$, so $p(x)^2=\Tr(pp^T M_x)$.  Accordingly, every sos polynomial $h$ of degree $\leq d$ corresponds to a psd matrix $Z$ such that $h(x)=\Tr(Z M_x)$ for all $x\in\{0,1\}^n$ ($Z$ can be written as $\sum_{i=1}^r p_ip_i^T$, so the rank~$r$ of $Z$ would be the number of squared polynomials that $h$ sums over). Hence the existence of an sos polynomial $h$ of degree $\leq d$ that $\eps$-approximates~$f$ in $\ell_\infty$-distance, is equivalent to the existence of a psd matrix $Z$ such that $|\Tr(ZM_x)-f(x)|\leq\eps$ for all $x\in\{0,1\}^n$. The latter corresponds to the feasibility of an SDP with $2^n$ constraints, which (up to issues of precision) can be solved in time $2^{O(n)}$. However, we won't use this fact in this paper.} 

For much of this paper we will focus on understanding the approximate sos degree of the symmetric quadratic functions 
$f_k: \{0,1\}^n \to \R_+$ defined as $f_k(x) = (|x| - k)(|x| - k - 1)$ for $k = 0, 1, 2, \ldots, n-1$.  
Our study of these functions is motivated by several reasons.  First, these functions have a close connection to 
the proof complexity of the knapsack problem \cite{G01}, and have recently been used to show lower bounds on 
semidefinite extension complexity \cite{LRS15}; we survey these two applications in \secref{sec:applications} below.
Furthermore, while the $f_k$ may look very special, they are not too far from a general symmetric 
quadratic polynomial with real coefficients that is nonnegative on the hypercube, for the following reason.  Any symmetric quadratic polynomial on the hypercube is of the form $p(|x|)$ for a quadratic univariate polynomial $p$.  The polynomial $p$ will have two roots.  If 
the roots are complex they must come in a conjugate pair and $p$ is already sos; if the roots are real, and not both 
$\le 0$ or $\ge n$, then they 
must lie in an interval $[k,k+1]$, for some $k \in \{0, 1, \ldots, n\}$, just as with $f_k$.  

In our first set of results, we give lower and upper bounds on the $\ell_\infty$-approximate sos degree of the functions 
$f_k$.

\begin{theorem}[$\ell_\infty$ sos approximations of $f_k$]
\label{thm:ell_infty}
	For all integers $n\geq 0$, $k \in \{1, \ldots, n-2\}$, and $\eps=\eps(n)$ satisfying $0 < \eps < 1/50$, we have 
	\begin{enumerate}
		\item $\sosdeg_\eps (f_k,\ell_\infty) = \Omega(\sqrt{k (n -k)})$
		\item $\sosdeg_\eps(f_k,\ell_\infty) =O(\sqrt{k(n-k)}+\sqrt{n\log(1/\eps)})$
	\end{enumerate}
\end{theorem}
\noindent 
We expect the lower bound can be improved for the case of small~$\eps$ to match the upper bound, but have been unable to show that so far.
Observe that in the case of constant error, we obtain the tight bound of $\sosdeg_{1/50}(f_k,\ell_\infty)= \Theta\left(\sqrt{k(n-k)}\right)$.  While we are not aware of any previous work on $\ell_\infty$-approximate sos degree, 
techniques of Grigoriev \cite{G01} can be used to show that, for $n$ odd, any degree-$(n-1)/2$ sos polynomial has error at least $\Omega(1/\log n)$ for approximating $f_{\floor{n/2}}$.  This derivation is given in \apref{ap:D}.

The similarity between our $\ell_\infty$ bounds for $f_k$ and Paturi's bound for the 0/1-valued functions $g_k$ defined above is 
striking.  For the upper bound, the connection can be seen as follows: we can construct an $\eps$-approximation to $f_k$ in $\ell_\infty$-distance by finding a univariate polynomial $e(z)$ such that $h(z)=(z-k)(z-k-1)+e(z)$ is globally nonnegative (i.e., $h(z)\geq 0$ for all $z\in\R$), and $|e(i)| \leq \eps$ on integer points $i \in \{0,1,\ldots, n\}$.  As $h(z)$ is a globally nonnegative \emph{univariate} polynomial, it is sos, and furthermore $h(|x|)$ is an $\eps$-approximation to $f_k$.  What are the 
requirements on~$e$?  It must be large enough to ``cancel out'' the negative values of $(z - k )(z - k - 1)$ in the interval 
$(k,k+1)$, but small on all integer points $0,1,\ldots,n$.  This is very similar to looking for an $\eps$-approximation 
of $g_k$, and techniques similar to those used by Paturi show that there is an $e$ satisfying these requirements of degree $O(\sqrt{k(n-k)} + \sqrt{n} \log(1/\eps))$.  Note that this is slightly weaker than what \thmref{thm:ell_infty} claims; we will soon discuss how to bring the $\log(1/\eps)$ inside the square-root.

This upper bound argument shows that $f_k$ can be approximated by a polynomial $h(|x|)$ where $h$ is globally 
nonnegative.  For the lower bound, it is not clear at all why the optimal approximating polynomial should be of this 
form.  Any symmetric polynomial $f(x_1, \ldots, x_n)$ on the hypercube is of the form 
$f(x_1, \ldots, x_n)=p(|x|)$ for a univariate polynomial $p$.  Even if $f$ is sos, however, this does not mean that $p$ 
will be globally nonnegative.  

For the lower bound, we use an elegant recent result of Blekherman~\cite{B15} that gives a characterization of the 
possible form of univariate polynomials $p$ such that $f(x_1, \ldots, x_n)=p(|x|)$ when $f$ is a sos and symmetric real-valued boolean function.  
This structural theorem allows us to reduce the analysis of the approximate sos degree of $f_k$ to the approximate 
degree of a symmetric function on the boolean hypercube, for which we can apply Paturi's lower bound.

Interestingly, for small $\eps$ we can show a better upper bound than that given by the argument sketched above.
To get the better upper bound of \thmref{thm:ell_infty}, we take advantage of a recent characterization of the sos degree of a non-negative real-valued boolean function as the quantum query complexity of computing that function \emph{in expectation} (see \secref{ssecqexp}). We explicitly design a quantum algorithm to approximately compute $f_k$ in expectation with query complexity $O(\sqrt{k(n-k)}+\sqrt{n\log(1/\eps)})$, which by the characterization implies the same upper bound on $\sosdeg_\eps(f_k,\ell_\infty)$.  This again parallels the situation for symmetric boolean-valued functions, where the tight upper bound of $O(\sqrt{k(n-k)} + \sqrt{n \log(1/\eps)})$ on $\deg_\eps(g_k)$ was first shown by the construction of a quantum query algorithm~\cite{deW08}.

We also study sos $\ell_1$-approximations of $f_k$:

\begin{theorem}[sos $\ell_1$-approximations of $f_k$]
\label{thm:l1_upperintro}
Let $n$ be odd and $k=\floor{n/2}$.  Then 
\[
\sosdeg_{\delta 2^n}(f_k,\ell_1) \le \left \lceil \frac{3\sqrt{n}}{\sqrt{2}\delta} \ln\left(\frac{1}{\delta} \right) \right \rceil 
\enspace ,
\]
for any $8/\sqrt{2n} \le \delta \le 1/4$.  For $k < 0.49 n$, we have
$\displaystyle
\sosdeg_{\delta 2^n}(f_k,\ell_1) = O\left(\ln \left(\frac{1}{\delta}\right)\right) \enspace.
$
\end{theorem}
The proof of this theorem follows the same plan sketched above for the upper bound in the $\ell_\infty$ case.  We 
construct a low-degree univariate polynomial $e(z)$ such that $h(z)=(z-k)(z-k-1)+e(z)$ is globally nonnegative and 
$\eps=\sum_{i=0}^n \binom{n}{i} |e(i)|$ is relatively small.  Then $h(|x|)$ gives the desired sos approximation to $f_k$ in $\ell_1$-norm.  We discuss the applications of this theorem to the lower bounds on semidefinite extension complexity 
of \cite{LRS15} below in \secref{sec:LRS}.


\subsection{Applications in complexity theory}
\label{sec:applications}

Here we describe complexity-theoretic consequences of such sos bounds in three different settings.

\subsubsection{Positive semidefinite extension complexity}
\label{sec:LRS}
The approximation of boolean functions by sos polynomials has played an important role in \emph{inapproximability} results.
Our first application is to the analysis of the positive semidefinite extension complexity of polytopes. The recent breakthrough work of Lee, Raghavendra and Steurer~\cite{LRS15} showed that any semidefinite program whose feasible region projects to the \emph{correlation polytope} must have size $2^{\tilde \Omega(n^{2/11})}$. By reduction this in turn implies a $2^{\tilde \Omega(n^{1/11})}$ lower bound for the polytope corresponding to the Traveling Salesman Problem on $n$-vertex graphs, showing (roughly speaking) that TSP cannot be solved by small semidefinite programs.

The argument of \cite{LRS15} shows that lower bounds on the degree of sos
polynomials that approximate a function $f_k(x)= (|x|-k)(|x|-k-1)$ in $\ell_1$-distance on the boolean cube imply lower 
bounds on the semidefinite extension complexity of the correlation polytope.  They build on the work of Grigoriev~\cite{G01} to show that, for odd $n$ and $k=\floor{n/2}$, any sum of squares of degree-$\floor{n/2}$ polynomials has 
$\ell_1$-error at least $2^{n-2}/\sqrt{n}$ in approximating $f_{\floor{n/2}}$.\footnote{The initial version of \cite{LRS15} only claimed a lower bound on the $\ell_1$-error of $\Omega(2^{n}/n^{3/2})$. However, their argument actually shows a bound of $\Omega(2^{n}/\sqrt{n})$ after a computational error is corrected.  This improves the bound on the psd extension complexity of the correlation polytope from the $2^{\tilde \Omega(n^{2/13})}$, of their paper, to the 
$2^{\tilde \Omega(n^{2/11})}$ quoted here.}
Our \thmref{thm:l1_upperintro} shows that this bound is tight, up to logarithmic factors.  Further, our upper bound 
on $\sosdeg_{\delta2^n}(f_k, \ell_1)$ throughout the full range of error implies, roughly speaking, that the quantitative 
bounds of \cite{LRS15} cannot be improved simply by showing better sos degree lower bounds on $f_k$.

\subsubsection{Proof complexity}\label{sec:proofcomplexity}
Our second result is in proof complexity.
Grigoriev and Vorobjov \cite{GV01} introduced a proof system based on the Positivstellensatz \cite{Ste74}. 
We explain this proof system in the context of the knapsack problem.  In this instance, the knapsack problem can be phrased as looking for a solution $x \in \R^n$ to the system of equations
\begin{equation}
\label{eq:knapsack_intro}
\sum_{i=1}^n x_i -r = 0,\, x_1^2-x_1=0, \ldots, x_n^2 - x_n =0
\end{equation}
where $r \in \mathbb{R}$. When $r$ is not an integer, this system obviously has no solution.  One way to certify that there is no solution, is to find polynomials $g, g_1, \ldots, g_n$ and sos polynomial $h$ such that 
\begin{equation}
\label{eq:knapsack_refutation}
g(x) \cdot \left(\sum_{i=1}^n x_i - r \right) + \sum_{i=1}^n g_i(x) \cdot (x_i^2-x_i) = 1 + h(x) \enspace.
\end{equation}
Such a collection of polynomials constitutes a \emph{Positivstellensatz refutation} of the statement 
that~\eqref{eq:knapsack_intro} has a solution: if $a \in \mathbb{R}^n$ satisfied $\sum_{i=1}^n a_i -r = 0$, and 
$a_i^2-a_i=0$ for $i=1, \ldots, n$, then the left-hand side of~\eqref{eq:knapsack_refutation} would evaluate to $0$ on~$a$, while the right-hand side would evaluate to $1 + h(a) \ge 1$, a contradiction. 

Grigoriev and Vorobjov define the \emph{Positivstellensatz refutation degree} of the system~\eqref{eq:knapsack_intro}
as 
\[
\max\left\{ \deg \left( g(x) \cdot \left(\sum_{i=1}^n x_i - r \right)\right), \max_{i \in [n]} \{\deg(g_i(x) \cdot (x_i^2-x_i))\}, \deg(h(x))\right\},
\]
maximized over all sets of polynomials satisfying~\eqref{eq:knapsack_refutation}.
Grigoriev~\cite{G01} shows that if $k < r < k+1$ for a nonnegative integer $k < (n-3)/2$, then any Positivstellensatz refutation of~\eqref{eq:knapsack_intro} has degree at least $2k+4$. We provide a simple proof of this in~\apref{ap:C} using Blekherman's theorem. Kurpisz et al.~\cite{KurpiszLeppanenMastrolilli16} independently also give an alternative proof of Grigoriev's lower bound, by showing a more general theorem that reduces the analysis of dual certificates for very symmetric sos proof systems (such as for knapsack) to the analysis of univariate polynomials.

Grigoriev's lower bound shows the weakness of the Positivstellensatz-based proof system: even to refute such easy instances it already needs polynomials of fairly high degree. We prove here that Grigoriev's lower bound is exactly tight, answering an open question from~\cite{G01}.

\begin{theorem}
\label{thm:knapsack_upper}
Let $k < r < k+1$ for a nonnegative integer $k$.  The Positivstellensatz refutation degree of~\eqref{eq:knapsack_intro} with this value of $r$ is at most $2k+4$.
\end{theorem}

\subsubsection{Quantum query complexity of approximating a function in expectation}\label{ssecqexp}
The third application is to quantum algorithms.
Kaniewski et al.~\cite{KLW14} observed a very close connection between the sos degree of a function $f:\B^n\to\R_+$ and a variant of quantum query complexity: $\sosdeg(f)$ is \emph{exactly equal} to the optimal query complexity among all quantum algorithms with non-negative outputs whose \emph{expected} output on input~$x$ equals $f(x)$.\footnote{To avoid potential confusion: for each fixed~$x$ the expectation is taken over the internal randomness of the algorithm; it is not an expectation over different inputs~$x$.} This model of query complexity in expectation is motivated by similar models of \emph{communication} complexity that arose in the study of extension complexity of polytopes~\cite{fmprw:tspj}.

However, as~\cite{KLW14} note, this model has some intrinsic interest and motivation as well. Suppose we want to approximate $F(x)=\sum_{i=1}^m f_i(x)$, where each $f_i$ is a non-negative function of $x\in\B^n$.  Then we can just compute, for each $i$, a random variable whose expected value is $f_i(x)$ and then output the sum of those random variables.  By linearity of expectation, the output will have the correct expectation~$F(x)$. It will be tightly concentrated around its expectation if the individual random variables have a variance that is not too large. Thus in some cases it suffices to compute the $f_i(x)$ in expectation only, rather than to compute the values $f_i(x)$ themselves (which may be much more expensive). In this example, it is actually not even necessary to compute each $f_i(x)$ \emph{exactly} in expectation.  If the $i$th random variable has an expectation that is within $\eps_i$ of $f_i(x)$, then the expected value of our output is within $\sum_{i=1}^m\eps_i$ of the correct value~$F(x)$. 

The same proofs that Kaniewski et al.~\cite{KLW14} used to equate $\sosdeg(f)$ and quantum query complexity in expectation also work in the approximate case.
For example, $\sosdeg_\eps(f,\ell_\infty)$ is the optimal query complexity among all quantum algorithms with non-negative outputs whose expected output on input~$x$ differs from $f(x)$ by at most $\eps$, for every $x\in\B^n$, and the analogous statements hold for approximation using the other norms.
Accordingly, our above results about approximate sos degree immediately translate to results about quantum query complexity of algorithms that approximate $f$ in expectation.

\subsection{Organization} The rest of the paper is organized as follows. In \secref{sec:ell_infty}, we prove our sos $\ell_\infty$-approximation bounds  (\thmref{thm:ell_infty}). In \secref{sec:ell_1} we prove upper bounds on the degree of sos $\ell_1$-approximations to $f_{\floor{n/2}}$ (\thmref{thm:l1_upperintro}), and show that the lower bound of~\cite{LRS15} on the extension complexity of the correlation polytope cannot be improved by obtaining better sos $\ell_1$-approximate degree lower bounds.  In \secref{sec:proof} we prove \thmref{thm:knapsack_upper}, showing tightness of Grigoriev's knapsack lower bound.

\section{Sum-of-squares approximation in $\ell_\infty$-norm}
\label{sec:ell_infty}
In this section we give lower and upper bounds on the $\ell_\infty$-approximate sos degree of the function 
$f_k(x) = (|x| - k)(|x| - k - 1)$ to prove \thmref{thm:ell_infty}.  

\begin{remark}
\label{rem:flip}
Throughout this section, we will assume that 
$k \le n/2$.  Letting $\bar x$ denote the bitwise complement of $x$, we see that 
$f_k(x) = (|\bar x| - (n-k-1))(|\bar x|-(n-k)) = (|\bar x| - \ell)(|\bar x| - \ell -1)$ where $\ell=n-k-1$.  Thus if $k > n/2$, then 
$\ell \le n/2$ and $f_k(x) = f_\ell(\bar x)$.  For any $h: \{0,1\}^n \rightarrow \R$ the functions $h(x)$ and $h(\bar x)$ have 
the same sos degree, so it suffices to work with $f_k$ for~$k \le n/2$.  
\end{remark}

\subsection{Lower bound preliminaries} \label{sec:lbp} 

The following lemma is implicit in Paturi~\cite{P92}.

\begin{lemma}[Paturi \cite{P92}]
\label{lem:paturi}
Let $p: \R \rightarrow \R$ be a univariate polynomial and suppose that $0 \leq p(i) \le c$ for all $i \in \{0,1,\ldots,n\}$.  If
$|p'(\alpha)| \ge \delta$ for some $0 \le \alpha \le n$, then $\deg(p) = \Omega(\tfrac{\delta}{c}\sqrt{\alpha(n-\alpha)})$.
\end{lemma}

We give a simple, but convenient, application of this lemma to the case where $p$ is bounded on $\{0,1,\ldots,n\}$, except 
possibly for a small set~$S$ near where $p$ is known to be small.

\begin{lemma}
\label{lem:paturi_app}
Let $p: \R \rightarrow \R$ be a univariate polynomial, $S \subseteq \{0,1,\ldots,n\}$, and suppose that the following bounds are known.
\begin{itemize}
\item $|p(i)| \le c$ for all $i \in \{0,1,\ldots,n\} \setminus S$, for a constant~$c$.  
\item $p(\alpha) \le \eps$, for some $\alpha \in \{0,1,\ldots,\floor{n/2}\}$.
\item $p(\beta) \ge a$ where $|\alpha - \beta| \le d_1$, for a constant $d_1$.
\item $\max_{i \in S} |i-\alpha| \le d_2$, for a constant $d_2$.
\end{itemize}
If $a\geq c > \eps$, then $\deg(p) = \Omega(\sqrt{\alpha(n-\alpha)})$, where the constant in the $\Omega(.)$ depends on $c,d_1,d_2$.
\end{lemma}

\begin{proof}
If $|p(i)|\le c$ for all $i \in S$, then applying Paturi's lemma directly we obtain a bound of
\[
\Omega\left(\frac{a-\eps}{d_1c}\sqrt{(\alpha-d_1)(n-\alpha+d_1)}\right) \enspace.
\]
Otherwise, suppose the maximum of $|p(i)|$ over $i \in \{0,1,\ldots,n\}$ is attained at $j \in S$.  Then the derivative of $p$ is at least 
$(|p(j)|-\eps)/d_2$.  Applying Paturi's lemma in this case gives a bound of 
\[
\Omega\left(\frac{|p(j)|-\eps}{d_2|p(j)|} \sqrt{(\alpha-d_2)(n-\alpha+d_2)}\right) \ge 
\Omega\left(\frac{1}{d_2}\left(1-\frac{\eps}{c}\right) \sqrt{(\alpha-d_2)(n-\alpha+d_2)} \right) \enspace.
\]
\end{proof}

We will also need the following elegant theorem of Blekherman \cite{B15}.  Recall that a \emph{symmetric} 
real-valued boolean function 
$f: \{0,1\}^n \rightarrow \R$ satisfies $f(x) = f(\pi(x))$ for all $x \in \{0,1\}^n$ and $\pi \in S_n$, where the permutation 
$\pi$ acts as $\pi((x_1, \ldots, x_n)) = (x_{\pi(1)}, \ldots, x_{\pi(n)})$.  For any symmetric boolean function $f$ of 
degree $d$, there is 
a univariate polynomial $\tilde f$ of degree $d$ such that $f(x_1, \ldots, x_n) = \tilde f(x_1+ \cdots + x_n)$.
  
\begin{theorem}[Blekherman \cite{B15}]
\label{thm:B}
Let $f:\{0, 1\}^{n} \to \R_{+}$ be a symmetric non-negative real-valued boolean function and $\tilde f$ a univariate polynomial 
such that $f(x_1, \ldots, x_n) = \tilde f(x_1 + \cdots + x_n)$. If $f$ can be written as the sum of squares of $n$-variate polynomials of degree $d \le n/2$, then we can write 
\al{ 
\tilde f(z) &= q_{d}(z) + z(n-z) q_{d-1}(z) + z(z-1)(n-z)(n-1-z) q_{d-2}(z) + \cdots \nl 
&\cdots + z(z-1)\cdots (z-d+1)(n-z)(n-1-z)\cdots(n-d+1-z)q_{0}(z)
}  
where each $q_{t}(z)$ is a univariate sos polynomial with $\sosdeg(q_t) \le t$.
\end{theorem}
\noindent In \apref{ap:B} we include a proof of Blekherman's theorem. In \apref{ap:C}, we use 
Blekherman's theorem to provide a simple proof of Grigoriev's lower bound \cite{G01} on the degree 
of Positivstellensatz refutations for the knapsack problem.  

\subsection{Lower bound for exact sos degree}

To illustrate our proof technique, we first show how the above tools can be used to prove a bound of $\Omega(\sqrt{k(n-k)})$
on the exact sos degree of $f(x) = (|x|-k)(|x|-k-1)$.  In the next section, we will extend this proof to also 
work for the approximate case.

\begin{theorem}
Let $f_k: \{0,1\}^n \to \R_+$ be defined as $f_k(x) = (|x|-k)(|x|-k-1)$ for an integer $1 \le k \le n-2$.  
Then 
$\sosdeg(f_k) = \Omega(\sqrt{k(n-k)})$.
\end{theorem}

\begin{proof}
Following \remarkref{rem:flip}, if we show the theorem for $1 \le k \le n/2$, then it will also imply the theorem for $1 \le k \le n-2$.  We thus assume $1 \le k \le n/2$.

Let $d$ be the sos  degree of $f$.  We may assume $d \le n/2$ as otherwise the theorem holds.  Let 
$\tilde f$ be a univariate polynomial of degree $\le 2d$ such that $\tilde f(x_1+\cdots + x_n)=f(x_1, \ldots, x_n)$.
Write $\tilde f(z) = g_1(z)+g_2(z)$ where 
$g_1(z) = q_d(z) + z(n-z)q_{d-1}(z) + \cdots + z(z-1)\cdots (z-(k-1))(n-z)(n-1-z)(n-(k-1)-z) q_{d-k}(z)$ is the first $k+1$ 
terms in the representation of $\tilde f$ of \thmref{thm:B}, and $g_2(z)$ is the remaining part of that representation.

Our first claim is that $(z-k)$ is a factor of both $g_1$ and $g_2$.
Notice that $\tilde f(k)=g_1(k)+g_2(k)=0$. Furthermore each term of $g_1$ and $g_2$ is nonnegative on integer 
points between $0$ and $n$, which means that each individual term of $g_1$ and $g_2$ must evaluate to $0$ at $k$.

Consider now a general term $z(z-1)\cdots (z-t)(n-z)(n-1-z)(n-t-z) q_{d-t-1}(z)$ of Blekherman's representation.  If $t \ge k$ then this term obviously 
has a factor of $z-k$.  If $t < k$ then the prefactor $z(z-1)\cdots (z-t)(n-z)(n-1-z)(n-t-z)$ is non-zero for $z=k$, so it must be the case that $q_{d-t-1}(k)=0$. Since $q_{d-t-1}(z)$ is a univariate sum-of-squares polynomial, even $(z-k)^2$ divides $q_{d-t-1}(z)$.  

By the choice of the breakpoint between $g_1$ and $g_2$, this shows that $(z-k)^2$ is a factor of $g_1$ and 
$z-k$ is a factor of $g_2$.  By the same argument, $(z-(k+1))^2$ is also a factor of $g_1$, and $(z-(k+1))$ is a factor of $g_2$.  

In light of this, we can write $g_1(z)=(z-k)(z-k-1) h_1(z)$, $g_2(z)=(z-k)(z-k-1) h_2(z)$ so that
\[
(z-k)(z-k-1) = \tilde f(z) = (z-k)(z-k-1)(h_1(z)+h_2(z)) \enspace .
\]
This means that $h_1(i)+h_2(i) = 1$ for all $i \in \{0,1,\ldots,n\} \setminus \{k,k+1\}$.  Furthermore, $h_1(k)=h_1(k+1)=0$ as $h_1$ 
still has roots at $k,k+1$ (as $g_1$ had double roots there), and $h_2(i)=0$ for $i \in \{0,\ldots,k-1\}$ because each term in $h_2$ 
includes the prefactor $z(z-1) \cdots (z-k+1)$.  Combining these observations with the fact that $h_1(i) \ge 0,h_2(i) \ge 0$ 
for all $i \in \{0,1,\ldots,n\} \setminus \{k,k+1\}$ gives the following:
\begin{enumerate}
  \item $ 0 \le h_1(i) \le 1$ for all $i \in \{0,1,\ldots,n\}$.
  \item $h_1(i) = 1$ for $i \in\{0,\ldots,k-1\}$.
  \item $h_1(k)=h_1(k+1) = 0$.
\end{enumerate}
Applying \lemref{lem:paturi_app} to $h_1$ now gives the desired result.
\end{proof}

\subsection{Lower bound for $\ell_{\infty}$-approximate sos degree}
Now we show the lower bound of \thmref{thm:ell_infty}.  
\begin{theorem}\label{th:lowerboundconstanterror}
\label{thm:lower}
Let $f: \{0,1\}^n \rightarrow \R_+$ be defined as $f(x) = (|x|-k)(|x|-k-1)$ for some integer $1 \le k \le n-2$.  Then 
$\sosdeg_{1/50}(f,\ell_\infty) = \Omega(\sqrt{k(n-k)})$.
\end{theorem}

\begin{proof}
We now describe how the above proof can be modified to work for $\ell_{\infty}$-approximate sum-of-squares degree.  
We again assume $1 \le k \le n/2$.  Suppose that 
$h : \{0,1\}^n \rightarrow \R$ is a sum of squares of degree-$d \le n/2$ polynomials that satisfies $|h(x)-f(x)| \le \eps$ for 
all $x\in \{0,1\}^n$, 
for some $\eps < 1/4$ to be determined later.  Let $\tilde h$ be the univariate polynomial such that 
$\tilde h(x_1 + \cdots + x_n)=h(x_1, \ldots, x_n)$. Note that $\tilde h$ satisfies 
$|\tilde h(i) - (i-k)(i-k-1)| \le \eps$ for all $i \in \{0,1,\ldots,n\}$.  We again use Blekherman's theorem to decompose 
$\tilde h(z)=g_1(z) + g_2(z)$ where, as before, 
$g_1(z) = q_d(z) + z(n-z)q_{d-1}(z) + \cdots + z(z-1)\cdots (z-k+1)(n-z)(n-1-z)(n-k+1-z) q_{d-k}(z)$.  The polynomial $g_1$ has the following properties.
\begin{enumerate}
\item $g_{1}(i) = \tilde h(i)$ for $i \in \{0,1,\ldots,k\}$, because all terms of $g_2$ are zero on these points.
\item $g_{1}(i) \leq \tilde h(i)$ for $i \in \{0,1,\ldots,n\}$.  This follows as $g_2(i)$ is nonnegative on integer points in $\{0,1,\ldots,n\}$.
\item $g_1(i) \ge 0$ for $i \in [k-1,n-k+1]$.  Each term of $g_1$ is nonnegative in this interval because the prefactor is.
\end{enumerate}
We will consider two cases based on the value of $g_1(k+3/2)$.  First consider the case $g_1(k+3/2) > \eps$.  In 
this case, consider a point $\alpha \in \argmin_z \{g_1(z): k-1 \le z \le k+3/2\}$.  Let $g_1(\alpha)=\delta$.
By item~(3) above and as $g_1(k-1), g_1(k+3/2) > \eps$ and $g_1(k), g_1(k+1) \le \eps$ we have 
$0 \le \delta \le \eps$ and also $g_1'(\alpha)=0$. 

Now consider the function $p_1=g_1-\delta$.  As $p_1(\alpha) = p_1'(\alpha)=0$ it follows that $p_1$ has a double root 
at $\alpha$.  Define $q_1$ by $p_1(z) = (z-\alpha)^2 q_1(z)$.  Note that $q_1$ has the following properties.
\begin{enumerate}
\item $q_1(i) \le 6 + \eps$ for $i \in \{0,1,\ldots,n\} \setminus \{k-1, k,k+1, k+2\}$.
\item $q_1(k-1) \ge \tfrac{2-2\eps}{9}$.
\item As either $|\alpha -k| \ge 1/2$ or $|\alpha - k-1| \ge 1/2$ we have either $q_1(k) \le 4\eps$ or 
$q_1(k+1) \le 4\eps$.  
\end{enumerate}
Applying \lemref{lem:paturi_app} then gives the desired lower bound in this case as long as $\eps < 1/19$. 

Now we consider the second case, that $g_1(k+3/2) \le \eps$.  In this case, we modify $g_1$ by adding a function that is shaped like a ``smile.''  Let 
$p_1(z)= g_1(z) + 8\eps (x-k-1)(x-k-2)$.  Note that $p_1$ satisfies 
$p_1(k+1) \ge 0$, $p_1(k+3/2) \le  -\eps$, and $p_1(k+2) \ge 0$.  Thus 
$p_1(z)$ has two roots $\alpha,\beta$ in $[k+1,k+2]$, with $\alpha\leq\beta$.  Let $p_1(z) = (z-\alpha)(z-\beta) r_1(z)$.  Then $r_1$ satisfies the following properties.
\begin{enumerate}
\item $r_1(i) \le 2$ for $i \in \{0,1,\ldots,n\} \setminus \{k+1, k+2\}$.
\item $r_1(k-1) \ge \frac{2}{9}+5\eps$.
\item $|r_1(k)|\le 16\eps$.
\end{enumerate}
Applying \lemref{lem:paturi_app} then gives the desired lower bound as long as $\eps < 2/99$. 
\end{proof}

\subsection{Upper bound for $\ell_\infty$-approximate sos degree}
In this section we show that the lower bound in \thmref{thm:lower} is tight.  To do this, we use the characterization of 
the sos degree of a function $f: \{0,1\}^n \rightarrow \R_+$ as the quantum query complexity of computing $f$ in expectation~\cite{KLW14}. In this model, a quantum algorithm $A$ makes a number $T$ of quantum queries to the hidden input $x$, and outputs a non-negative real number. We say that the algorithm $A$ \emph{computes $f$ in expectation} if the expected value of the output of the algorithm $A$ on input $x$ is exactly equal to $f(x)$. We will use $\QE(f)$ to denote the minimum number of quantum queries $T$ needed by such an algorithm to compute $f$ in expectation. Kaniewski et al.~\cite{KLW14} show that $\QE(f)$ exactly captures the sos degree of $f$:

\begin{theorem}[\cite{KLW14}]
\label{thm:klw}
Let $f : \{0,1\}^n \to \R_+$. Then $\QE(f) = \sosdeg(f)$.
\end{theorem}
Thus, in order to prove an upper bound on the (approximate) sos degree of a function $f$, it suffices to construct a quantum query algorithm that (approximately) computes $f$ in expectation. The only knowledge of quantum query complexity needed to understand the algorithm is \thmref{thm:klw} above, and the existence of the following quantum algorithms, all of which are variants of Grover search.
\begin{itemize}
\item \emph{Regular Grover}~\cite{grover:search,bhmt:countingj}: If $|x| \ge t$ then there is a quantum algorithm 
(depending on $t$) using $O(\sqrt{n/t})$ queries that finds an $i$ such that $x_i=1$ with probability at least $1/2$. 
\item \emph{$\eps$-error Grover}~\cite{bcwz:qerror}: There is a quantum algorithm using $O(\sqrt{n\log(1/\eps)})$ queries that finds an $i$ such that $x_i=1$ with probability at least $1-\eps$ if $|x| \ge 1$. 
\item \emph{Exact Grover}~\cite{bhmt:countingj}: If $|x|=t$ then there is a quantum algorithm (depending on $t$) using $O(\sqrt{n/t})$ queries that finds an $i$ such that $x_i=1$ \emph{with certainty}. 
\end{itemize}

The algorithm consists of three subroutines, which we now describe.  We begin with the simplest procedure, SAMPLE$(x,S)$, which motivates the basic plan of the algorithm.

\begin{algorithm}
\caption{Given $x \in \{0,1\}^n$ and $S \subseteq [n]$, samples two entries of $x$ outside of $S$}
\begin{algorithmic}[1]
\Procedure{Sample}{$x,S$}
\State Randomly choose $i \ne j \in [n] \setminus{S}$.  Output $x_i x_j \cdot (n-|S|)(n-|S|-1)$
\EndProcedure
\end{algorithmic}
\label{alg:sample}
\end{algorithm}

\begin{claim}
\label{claim:sample}
The procedure SAMPLE$(x,S)$ makes two queries and the expected value of its output is $(|x|-|S|)(|x|-|S|-1)$. 
\end{claim}

The procedure SAMPLE suggests the following high-level idea for an algorithm for computing $f_k(x)=(|x|-k)(|x|-k-1)$.  
For simplicity we describe the 
high-level idea for the case where $k \le n/2$ and where we want to compute a constant-error $\eps$-approximation of 
$f_k$, in expectation.  

First we try to find a set $S$ of $k$ ones in $x$ assuming that $|x| > 2k$, using a procedure HIGH.  If we find such a set 
$S$ then we run SAMPLE$(x,S)$ and output $f(x)$ exactly, in expectation.  If the procedure HIGH fails to find such a set $S$, then we 
run a procedure LOW.  This uses exact Grover search to determine the Hamming weight of $x$ with certainty if 
$|x| \le 2k$.  Once we know the Hamming weight of $x$ we can correctly output $f(x)$, deterministically.  Both the 
procedures HIGH and LOW can be done with $O(\sqrt{kn})$ queries, in the constant error $\eps$ case.
The only case where the algorithm may err is if $|x| > 2k$ but the procedure HIGH fails to find $k$ ones in $x$.  The most subtle part of the algorithm is tuning the parameters such that this error is at most $\eps$ in 
expectation.  We now describe the procedures HIGH and LOW.

\begin{algorithm}
\caption{Find $k$ ones in $x$ with probability $1-\delta$, assuming $|x| \ge t > 2k$}
\begin{algorithmic}[1]
\Procedure{High}{$x,t,\delta$}
\State $S = \emptyset$
\State $\ell =1$
\While{$\ell \le 5\max(k,\lceil\log(1/\delta)\rceil)$ and $|S| < k$}
\State $\ell \gets \ell+1$
\State Grover search assuming $|x| \ge t/2$. 
\If{find $x_i=1$} $S \gets S \cup {i}$, $x \gets x \setminus x_i$
\EndIf
\EndWhile
\State \Return $S$
\EndProcedure
\end{algorithmic}
\end{algorithm}

\begin{lemma}
\label{lem:high}
Fix $\delta$ and let $M=\max\{k, \lceil\log(1/\delta)\rceil\}$.  Suppose that $|x| \ge t > 2k$.  Then procedure HIGH($x,t,\delta$) 
makes $O(M\sqrt{n/t})$ queries and returns a set $S$ with $|S|=k$ and $x_i=1$ for all $i \in S$ with probability at least $1-\delta$.  
\end{lemma}

\begin{proof}
As each Grover search takes $O(\sqrt{n/t})$ queries, in total the 
procedure makes $O(M\sqrt{n/t})$ queries.  Let us now estimate the probability that it exits without finding 
a set $S$ of size $k$.

As we are given that initially $|x|\ge t > 2k$, if less than $k$ ones are found then throughout the algorithm there remain 
at least $t/2$ ones in $x$.  Thus each run of Grover has probability of success at least 
$1/2$.  The probability to have fewer than $k$ successes among the $5M$ runs is therefore at most
\[
\frac{1}{2^{5M}}\sum_{i=0}^{k-1}{5M\choose i}\leq 2^{-(1-H(k/5M))5M}\leq 2^{-M}\leq 2^{-\log(1/\delta)}=\delta,
\]
where $H(\cdot)$ denotes binary entropy, and we used that $1-H(k/5M)\geq 1-H(1/5)\geq 1/5$.
\end{proof}

Next we give the algorithm LOW.

\begin{algorithm}
\caption{Outputs $(|x|-k)(|x|-k-1)$ with certainty if $|x| \le t$}
\begin{algorithmic}[1]
\Procedure{Low}{$x,t$}
\State $S=\emptyset$
\For{$i=t$ to $1$}
\State Exact Grover search assuming $|x|=i$
\If{find $x_i=1$} $S \gets S \cup {i}$, $x \gets x \setminus x_i$
\EndIf
\EndFor
\State Output $(|S|-k)(|S|-k-1)$.
\EndProcedure
\end{algorithmic}
\end{algorithm}

\begin{claim}
\label{claim:low}
If $|x| \le t$, then LOW(x,t) outputs $(|x|-k)(|x|-k-1)$ and makes $O(\sqrt{tn})$ queries.  
\end{claim}

\begin{proof}
The number of queries is
\[
\sum_{i=1}^t O(\sqrt{n/i}) = O(\sqrt{tn}) \enspace.
\]
Next we show that if $|x| \le t$, then LOW$(x,t)$ will find all of the ones in $x$ (this is similar to~\cite{graaf&wolf:qyao}).  Initially the index $i=t$ and thus 
$i \ge |x|$.  This invariant is maintained throughout the algorithm.  If ever $i=|x|$ then we will find all the remaining 
ones in $x$ as our guess for the number of ones is always correct after this point.  On the other hand, if the algorithm 
terminates with $i=1>|x|$ then we have found all the ones in the original input $x$.
\end{proof}

With these procedures in place, we can describe the main algorithm and prove its correctness.
\begin{algorithm}
\caption{Main}
\begin{algorithmic}[1]
\Procedure{Main}{$x,\eps$}
\State $m=\max(k,\ceil{\log(1/\eps)})$
\For{$i=1$ to $\floor{\log (n/m)}$}
\State $t \gets 2^i m$
\State $\delta \gets \eps/(4t^2)$
\State $S$=HIGH$(x,t,\delta)$
\If{$|S|= k$} 
\State SAMPLE$(x,S)$
\State Exit
\EndIf
\EndFor
\State LOW$(x,2m)$
\EndProcedure
\end{algorithmic}
\label{alg:main}
\end{algorithm}

\begin{theorem}
\label{thm:alg}
For every $x \in \{0,1\}^n$, the expected value of Main$(x,\eps)$ differs from 
$(|x|-k)(|x|-k-1)$ by at most $\eps$.  The algorithm makes at most 
$O(\sqrt{kn} + \sqrt{n \log(1/\eps)})$ queries.  
\end{theorem}

\begin{proof}
Following \remarkref{rem:flip} we may assume that $k \le n/2$.  First we verify the stated complexity of the algorithm.  Note 
that by definition of $m$ in the main \algoref{alg:main}, 
it suffices to show that the 
algorithm makes $O(\sqrt{nm})$ queries.  By \claimref{claim:low} the call to LOW$(x,2m)$ makes 
$O(\sqrt{mn})$ queries, and by \claimref{claim:sample} there are at most 2 queries made by SAMPLE as this is called at most 
once.  Finally, the number of queries in the call to HIGH when $t=2^i m$ and $\delta = \eps/(4t^2)$  is at most
\[
O\left( k\sqrt{\frac{n}{2^i m}} + \log(2^{2i+2}m^2/\eps)\sqrt{\frac{n}{2^i m}}\right)
=O\left( \sqrt{\frac{kn}{2^i}} + \sqrt{\frac{n \log(1/\eps)}{2^i}} + \log(2^{2i+2} m^2) \sqrt{\frac{n}{2^i m}} \right)
\]
where we have used the fact that $m \ge k$ and $m \ge \log(1/\eps)$.  The sum of the first two terms over $i \ge 1$ is
$O(\sqrt{kn} + \sqrt{n \log(1/\eps)})$ as desired.  As for the sum of the third term, we have 
\[
\sum_{i \ge 1} O\left( \log(2^{2i+2} m^2) \sqrt{\frac{n}{2^i m}} \right) = O \left( \log(m) \sqrt{\frac{n}{m}} \right) = O(\sqrt{n}) 
\enspace .
\]
We now verify correctness.  If $|x| \le 2m$ then the algorithm will output $(|x|-k)(|x|-k-1)$ in expectation exactly: if $k$ 
ones are found in $x$ by a call to HIGH then this will be done by SAMPLE, otherwise all ones in $x$ will be found 
with certainty by LOW, which will then output correctly.  If $|x| > 2m$ and a call to HIGH succeeds in finding $k$ ones in 
$x$, the algorithm will also output $(|x|-k)(|x|-k-1)$ exactly, in expectation.  Let $p$ be the probability that this does not 
happen, i.e., that the output on $x$ is given by the procedure LOW.  Then the expected value of the output on $x$ is 
\[
(1-p) (|x|-k)(|x|-k-1) +p \cdot \Ex[\mathrm{LOW}(x,2m)] 
 \enspace,
\]
and the deviation from the desired output $(|x|-k)(|x|-k-1)$ is
\[
p \cdot (\Ex[\mathrm{LOW}(x,2m)] - (|x|-k)(|x|-k-1)) \enspace .
\]
Now $\mathrm{LOW}(x,2m)$ will always output a value $0 \le (\ell-k)(\ell-k-1)$ for some $\ell \in [2m]$, which is always 
at most the correct value $(|x|-k)(|x|-k-1)$ as $|x| \ge 2 m > k$.  Therefore the largest difference between these 
is when $\mathrm{LOW}(x,2m)$ outputs $0$, giving
\[
|p \cdot (\Ex[\mathrm{LOW}(x,2m)] - (|x|-k)(|x|-k-1))| \le p\cdot (|x|-k)(|x|-k-1) \le p\cdot |x|^2 \enspace .
\]
We now finally upper bound this error by giving an upper bound on $p$.  
Let $i$ and $t=2^i m$ be such that $t \le |x| < 2t$.  For this value of $t$ and $\delta=\eps/(4t^2)$ the call to 
HIGH$(x,t,\delta)$ fails to find a set $S$ of size $k$ with probability at most $\delta \le \eps/ |x|^2$.  Thus 
$p\cdot |x|^2 \le \delta\cdot |x|^2 \le\eps$, as desired.
\end{proof}

By \thmref{thm:klw}, the characterization of sos degree in terms of quantum query complexity in expectation (\thmref{thm:alg}) gives the upper bound in~\thmref{thm:ell_infty}


\section{Sum-of-squares approximation in $\ell_1$-norm}
\label{sec:ell_1}
In this section, we show upper bounds on the sos degree of polynomials to approximate $f_k$ in $\ell_1$-norm. In this section we focus 
on the case where $k$ is $\floor{n/2}$.  
 When $k < 0.49n$ the function $f_k$ is quite easy to approximate in $\ell_1$-norm: there is an sos polynomial of degree $O(\ln(1/\delta))$ which gives a $\delta 2^n$-approximation. We omit the details.%
\footnote{One way to construct such an sos polynomial is to construct a polynomial $e$ as mentioned after \thmref{thm:l1_upperintro} from a classical sampling algorithm: query $O(\ln(1/\delta))$ randomly chosen input bits; output some large number if the observed ratio of 1s is very close to $k/n$, output~0 otherwise. This induces an sos polynomial with the right properties.}  Our main result on $\ell_1$-approximation is the following.
\begin{theorem}
\label{thm:l1_upper}
Let $n$ be odd and $k=\floor{n/2}$. Then for any $8/\sqrt{2n} \le \delta \le 1/4$
\[
\sosdeg_{\delta 2^n}(f_k,\ell_1) \le \left \lceil \frac{3\sqrt{n}}{\sqrt{2}\delta} \ln\left(\frac{1}{\delta} \right) \right \rceil \enspace .
\]
\end{theorem}

Lee, Raghavendra, and Steurer~\cite{LRS15}, building on work of Grigoriev~\cite{G01}, show that in this case 
$\sosdeg_{2^n/\sqrt{n}}(f, \ell_1) \ge (n-1)/2$.  This lower bound was then plugged into their general theorem to 
lift $\ell_1$-approximate sos degree lower bounds to lower bounds on semidefinite extension complexity.  
By taking $\delta = 3\ln(n)/\sqrt{2n}$, \thmref{thm:l1_upper} shows that this lower bound on the 
$\ell_1$-error is tight, up to a logarithmic factor.  Also, taking $\delta$ to be a small additive constant shows that there is a 
degree-$O(\sqrt{n})$ sos polynomial which, on average, disagrees with $f_k$ by only a small constant.  Taken as a whole, 
\thmref{thm:l1_upper} implies that the quantitative bounds on the semidefinite extension complexity of 
the correlation polytope of \cite{LRS15} cannot be improved simply by improving the sos degree lower bounds on the $f_k$.  We now describe the connection to \cite{LRS15} in greater detail.

\subsection{The theorem of Lee, Raghavendra, and Steurer}
For a function $f: \{0,1\}^n \rightarrow \mathbb{R}$ and an integer $N \ge n$, let 
$M_N^f: \binom{N}{n} \times \{0,1\}^N \rightarrow [0,1]$ be the matrix where $M_N^f(S,x)=f(x|_S)$.  The 
\emph{pattern matrix} of $f$, introduced in the work of Sherstov \cite{She09b}, is a submatrix of $M_N^f$.
The main theorem of Lee, Raghavendra, and Steurer is the following statement.
Here a \emph{degree-$d$ pseudo-density} $D$ is a function $D:\{0,1\}^n \rightarrow \mathbb{R}$ such that $ \Ex_x[D(x)] = 1$ and $\Ex_x[D(x) g(x)^2] \geq 0$ for all polynomials $g$ of degree at most $d/2$ on the boolean cube, with the expectation over a uniformly random $x \in \B^n$. We use $\|D \|_\infty$ to denote $\max_{x \in \B^n} |D(x)|$.

\begin{theorem}[\cite{LRS15}]
\label{thm:LRS}
Let $f: \{0,1\}^n \rightarrow [0,1]$.
  If there exists an 
$\eps \in (0,1]$ and a degree-d pseudo-density $D:\{0,1\}^n \rightarrow \mathbb{R}$ satisfying 
$\mathbb{E}_x [D(x) f(x)] < -\eps$, then for every $N \ge 2n$
\[
\psdrk(M_N^f) \ge \left( \frac{c \eps N}{dn^2 \|D\|_\infty \log n}\right)^{d/4} 
\left(\frac{\eps}{\|D\|_\infty} \right)^{3/2} \sqrt{\mathbb{E}_x f(x)} \enspace ,
\]
where $c > 0$ is a universal constant.
\end{theorem}

We do not formally define here the positive semidefinite (psd) rank of a matrix (denoted by $\psdrk$ above), but remark that psd rank lower bounds are equivalent to semidefinite extension complexity lower bounds. Lee, Raghavendra, and Steurer proved \thmref{thm:LRS} en route to their breakthrough result on superpolynomial size lower bounds for semidefinite programming relaxations of hard optimization problems. 

The more pertinent aspect of \thmref{thm:LRS} to us is the role of the degree-$d$ pseudo-density~$D$.  Note that, 
once we fix the degree of the pseudo-density, the bound only depends on the ratio 
$\Ex_x [D(x) f(x)] / \|D\|_\infty$.  The largest 
such ratio that a degree-$d$ pseudo-density can achieve is closely related to the best $\ell_1$-approximation of $f$ by degree-$d/2$ sos polynomials.
The following claim follows from strong duality of semidefinite programming.

\begin{claim}
\label{clm:ell_1_duality}
Let $f : \{0,1\}^n \rightarrow \mathbb{R}$.  Then $\sosdeg_{\delta 2^n}(f,\ell_1) >d$ if and only if there exists a ``witness'' 
function $\psi: \{0,1\}^n \rightarrow \R$ satisfying $\Ex_x[f(x) \psi(x)] > \delta$, and $\Ex_x [ p^2(x) \psi(x)] \le 0$ for all 
polynomials $p$ of degree at most $d$, and $\|\psi\|_\infty = 1$.  
\end{claim}
For $n$ odd and $k = \floor{n/2}$ Lee et al.\ \cite{LRS15}, building on work of Grigoriev \cite{G01}, show that there is a 
degree-$(n-1)$ pseudo-density~$D$ such that $\Ex_x [D(x) f_{\floor{n/2}}]/\|D\|_\infty < -\tfrac{1}{4\sqrt{n}}$.  Plugging 
$f =f_{\floor{n/2}}/n^2$ (with this normalization the range is in $[0,1]$) into \thmref{thm:LRS} gives a lower bound of
$2^{\widetilde{\Omega}(N^{2/11})}$ on the psd rank of $M_N^{f}$ for $N=\widetilde{O}(n^{11/2})$.  As $M_N^{f}$ is a 
submatrix of the slack matrix of the correlation polytope, this gives the desired lower bound on the semidefinite extension 
complexity of the correlation polytope. 

In light of \claimref{clm:ell_1_duality}, if $\sosdeg_{\delta 2^n}(f,\ell_1) \le d$, then there can be no degree-$2d$ 
pseudo-density with $\Ex_x [D(x) f(x)] / \|D\|_\infty < - \delta$.  The $\ell_1$-approximate sos 
degree upper bounds of \thmref{thm:l1_upper} therefore imply the non-existence of pseudo-densities with good 
properties for \thmref{thm:LRS}.  It 
can be verified that $2^{\widetilde{\Omega}(N^{2/11})}$ is in fact the best quantitative bound that 
\thmref{thm:LRS} can show on $\psdrk(M_N^{f_k})$ over all the functions $f_k$ and tradeoffs between $\delta$ and 
$\sosdeg_{\delta 2^n}(f_k,\ell_1)$.

\subsection{Proof of~\thmref{thm:l1_upper}}
Throughout this proof we set $f=f_{\floor{n/2}}$.
The main idea of the proof of~\thmref{thm:l1_upper} is to construct a univariate polynomial $p$ such that 
$h(z)=(z-\floor{n/2})(z-\ceil{n/2})+p(z)$ is 
globally nonnegative (and therefore sos) and $\sum_{i=0}^n \binom{n}{i} |p(i)|$, which is the $\ell_1$-error of $h(|x|)$ in approximating $f$, is 
reasonably small.  We will construct $p$ using Chebyshev polynomials.  Similar constructions to what we need have 
been done before, see for example \cite{She09};  as our requirements are somewhat specific, however, we do the 
construction from scratch.

Let $T_d$ be the Chebyshev polynomial of 
degree $d$.  We first recall some basic facts about Chebyshev polynomials \cite{Riv69}.
\begin{fact} 
\label{fac:cheb}
Let $T_d(z)$ be the Chebyshev polynomial of degree $d$.  Then
\begin{enumerate}
\item $|T_d(z)| \le 1$ for $z \in [-1,1]$. 
\item $T_d(z) = \tfrac{1}{2}\left( (z-\sqrt{z^2-1})^d + (z+\sqrt{z^2-1})^d \right)$.
\item $T_{d+1}(z) = 2zT_{d}(z) - T_{d-1}(z)$.
\item $T_d(z)$ is monotonically increasing for $z \ge 1$, and if $d$ is even $T_d(z) \ge 1$ for 
$z \in \R \setminus [-1,1]$.
\end{enumerate}
\end{fact}

\begin{theorem}
\label{thm:middle_const}
Let $n \ge 1$ be an integer, $\eps \in (0,1/4]$ an error parameter, and let $a \in \mathbb{R}$ satisfy
$1/\sqrt{2} \le a \le \sqrt{n}/8$.  There is a polynomial $p$ of degree 
at most 
\[
\left \lceil \frac{3n}{4\sqrt{2}\,a} \ln\left (\frac{1}{2\eps}\right) \right  \rceil +1 \enspace
\] 
with the following properties:
\begin{enumerate}
\item $p(z) \ge \tfrac{1}{4}-z^2$ for all $z$.
\item $|p(z)| \le \eps$ for $z \in [-\tfrac{n}{2}, -a] \cup [a, \tfrac{n}{2}]$.
\item $|p(z)| \le 2$ for $z \in [-\tfrac{n}{2}, -\tfrac{1}{2}] \cup [\tfrac{1}{2}, \tfrac{n}{2}]$.
\end{enumerate}
\end{theorem}

\begin{proof}
Note that we require in particular that $p(0) \ge 1/4$.  Roughly speaking, $p$ should have a `peak' around $0$ and 
then quickly calm down and be bounded on either side of this peak once $|z| \ge a$.  
The difficulty in constructing $p$ is that its peak 
is \emph{in between} the intervals on which is bounded.  To get around this, we note that is suffices to let
$p(z) = \eps q(z^2)$, where $q$ has the properties
\begin{enumerate}
\item $q(z) \ge \tfrac{1}{4\eps}-\tfrac{z}{\eps}$ for $z \ge 0$.
\item $|q(z)| \le 1$ for $z \in [a^2, \tfrac{n^2}{4}]$.
\item $|q(z)| \le \tfrac{2}{\eps}$ for $z \in [\tfrac{1}{4},\tfrac{n^2}{4}]$.
\end{enumerate}

Now $q(z)$ is ripe for a construction with Chebyshev polynomials, and this is what we do.  For notational convenience, 
let $L=n/2$.  Define the mapping $s(z)=-2(z-a^2)/(L^2-a^2)+1$ that takes the interval $[a^2, L^2]$ to $[-1,1]$.
Note that this mapping takes $L^2$ to $-1$ and $a^2$ to $1$.  
Let $T_d$ be the Chebyshev polynomial of \emph{even} degree $d$ (to be chosen later) and define
\[
q(z) = T_d(s(z)) \enspace .
\]  
As $|T_d(z)| \le 1$ for $z \in [-1,1]$ by \factref{fac:cheb} item~(1), it follows that $q(z)$ satisfies condition~$(2)$.  

We now turn to item~(1) and handle the easy cases first.  For $z \ge \tfrac{1}{4}$ we have 
$\tfrac{1}{4\eps}-\tfrac{z}{\eps} \le 0$, so in this region we just need to check that $q(z)$ is not too negative.  
If $z \in [\tfrac{1}{4},a^2]$, then $s(z) \ge 1$ and therefore $q(z) \ge 1$.  Likewise, as we take $d$ to be even, 
$q(z) \ge 1$ for $z \ge L^2$.  For $z \in [a^2,L^2]$, we have $|q(z)| \le 1$.  Thus item~(1) will be satisfied in this region 
so long as $a^2 \ge \tfrac{1}{4}+\eps$.  This holds as in the theorem statement $\eps \le 1/4$ and $a \ge 1/\sqrt{2}$.  

With these easy cases taken care of, we turn to verify the first item for $z \in [0,\tfrac{1}{4}]$.  To do this it suffices to 
choose $d$ such that $q(1/4) \ge \tfrac{1}{4\eps}$ as $q(z)$ is monotonically decreasing in the interval $[0,1/4]$, since 
$T_d(y)$ is monotonically increasing for $y \ge 1$ by \factref{fac:cheb} item~(4).  
This condition is at odds with item~(3).  As the maximum of $q(z)$ in the interval $[1/4,L^2]$ is attained at $z=1/4$, we 
can simultaneously satisfy item~(3) by ensuring $q(1/4) \le \tfrac{2}{\eps}$.  Thus we choose $d=d^*$ to be the least 
even number such that 
\[
q(1/4) = T_{d^*}\left( 1 + \frac{2(a^2-1/4)}{L^2-a^2}\right) \ge \frac{1}{4\eps} \enspace.
\]
By this choice, item~(1) is now satisfied.  To verify item~(3), we use \factref{fac:cheb} item~(3) to see 
the inequality $T_{s+2}(z) \le 4z^2 T_s(z)$, valid for $z \ge 1$.  Applying this we have
\[
T_{d^*}\left( 1 + \frac{2(a^2-1/4)}{L^2-a^2}\right) \le \frac{1}{\eps} \left( 1 + \frac{2(a^2-1/4)}{L^2-a^2}\right)^2 
\le \frac{2}{\eps} \enspace,
\]
as $T_{d^*-2}\left( 1 + \frac{2(a^2-1/4)}{L^2-a^2}\right) < \frac{1}{4\eps}$ by definition, and $a \le \sqrt{n}/8$.

Finally, we upper bound $d^*$.  Let $\mu = 2(a^2-1/4)/(L^2-a^2)$.  We want $T_d(1+\mu) \ge \tfrac{1}{4\eps}$.  
Using the fact that $T_d(1+\mu) \ge (1/2)(1+\sqrt{2\mu})^d$ for $\mu \ge 0$ by \factref{fac:cheb} item~(2), it suffices 
to take $d \ge \ln(\tfrac{1}{2\eps})/\ln(1+\sqrt{2\mu})$.  

As $\ln(1+y) \ge 2y/(2+y)$ for $y \ge 0$ and $\sqrt{2\mu} \le 1$, it suffices to take 
$d \ge 3 \ln(\tfrac{1}{2\eps})/(2\sqrt{2\mu})$. 
Since $a \ge 1/\sqrt{2}$, and therefore $a^2-1/4 \ge a^2/2$, we have $\mu\geq a^2/L^2=4a^2/n^2$. Hence there is a $d$ such that $T_d(1+\mu) \ge \tfrac{1}{4\eps}$ satisfying
\[
d \le \left \lceil \frac{3n}{4\sqrt{2}\,a} \ln\left (\frac{1}{2\eps}\right) \right  \rceil \enspace.
\] 
We add $1$ in the theorem statement for the additional requirement that the degree is even.
\end{proof}

\begin{proof}[Proof of \thmref{thm:l1_upper}]
Fix $8/\sqrt{2n} \le \delta \le 1/4$, and let $\eps =\delta/2$ and $a=\eps \sqrt{n}/4$.  
Note that $1/\sqrt{2} \le a \le \sqrt{n}/8$ with these 
choices.  Thus by \thmref{thm:middle_const}, there is a polynomial $p$ of degree at most 
$\ceil{\tfrac{6\sqrt{n}}{\sqrt{2}\delta}\ln(\tfrac{1}{\delta})}+1$ satisfying the three conditions of \thmref{thm:middle_const} 
with this value of $a, \eps$.  

Let $g(z) = (z-n/2)^2-1/4 +p(z-n/2)$ be a univariate polynomial, and consider the approximation to 
$f$ given by $g(|x|)$.  By construction $g$ is globally nonnegative and thus (as it is univariate)
is a sum of squares of polynomials of degree at most $\ceil{\tfrac{3\sqrt{n}}{\sqrt{2}\delta}\ln(\tfrac{1}{\delta})}$.  
Let us examine the $\ell_1$-error of the function $g(|x|)$ in 
approximating $f$.  We divide the error into two cases: the error on strings whose Hamming weight is at most 
$n/2-a$ or at least $n/2+a$ (type I), and those whose Hamming  weight is in the interval $[n/2-a,n/2+a]$ (type II).  

As $p$ is bounded by $\eps$ for $z \in [-n/2,-a] \cup [a, n/2]$ the $\ell_1$-error over type I inputs is at most 
$\eps \cdot 2^n$.  The number of type II inputs is at most $(2a/\sqrt{n})2^n$, and the 
error on each is at most $2$ as $p(z) \le 2$ for $z \in [-n/2, n/2]$.  Thus the total $\ell_1$-error is 
$\displaystyle
2^n \left( \eps + \frac{4a}{\sqrt{n}}\right)=2^n \cdot 2\eps = \delta 2^n.
$
\end{proof}


\section{Proof complexity: Positivstellensatz refutations}
\label{sec:proof}

Say that we have a system of polynomial equalities 
\begin{equation}
\label{eq:sys}
f_1 = \cdots = f_m = 0, \ x_1^2-x_1 = \cdots = x_n^2-x_n =0
\end{equation}
where each $f_i \in \mathbb{R}[x_1, \ldots, x_n]$.  Because of the presence of the 
equalities $x_i^2 -x_i=0$ (which force $x_i\in\{0,1\}$), this is referred to as the \emph{boolean} setting.  

The Positivstellensatz~\cite{Ste74} implies that the system~(\ref{eq:sys}) has no common solutions in $\mathbb{R}^n$ if 
and only if there are polynomials $g_1, \ldots, g_{m+n} \in \mathbb{R}[x_1, \ldots, x_n]$ and a sos
polynomial $h$ such that 
\begin{equation}
\label{eq:ref}
\sum_{i=1}^m f_i g_i  + \sum_{i=1}^n (x_i^2-x_i) g_{m+i}=1+h \enspace .
\end{equation}
Grigoriev and Vorobjov \cite{GV01} define a proof system based on this principle.  
\begin{defn}
A \emph{Positivstellensatz refutation} 
of the system~(\ref{eq:sys}) is given by a set of polynomials $\{g_1, \ldots, g_{m+n},h\}$ which satisfy (\ref{eq:ref}) and 
where $h$ is a sum of squares.  The \emph{degree} of this refutation is 
\[
\max\{\deg(h), 
\max_{i \in [m]} \deg(f_i g_i), \max_{i \in [n]} \deg((x_i^2-x_i)g_{m+i})\} \enspace.
\]  
\end{defn}

By the Positivstellensatz, this proof system is sound and complete: a system is unsatisfiable if and only if it has a 
refutation of a certain degree.  One may view the degree of a refutation as a measure of complexity.

\subsection{Knapsack}
The knapsack system is given by the equations
\begin{equation}
\label{eq:knapsack}
f=\sum_i x_i -r =0, \ x_j^2 - x_j = 0 \mbox{ for } j=1, \ldots,n \enspace .
\end{equation}
If $r$ is not an integer then this system has no solution: Grigoriev \cite{G01} shows the following theorem.

\begin{theorem}[Grigoriev \cite{G01}]
\label{thm:grig}
Let $0 \le k \le (n-3)/2$ be an integer.  If $k < r < n-k$, then any Positivstellensatz refutation of the 
system~(\ref{eq:knapsack}) has degree at least $2k+4$.
\end{theorem}
We provide a simple proof of this in \apref{ap:C} using Blekherman's theorem.

Note that the equations for non-integer $r$ correspond to a trivially easy (and obviously unsatisfiable) instance of the 
knapsack problem, where all items have weight~1.  As mentioned in the introduction, this shows the weakness of the Positivstellensatz-based 
proof system: even to refute such easy instances it already needs polynomials of fairly high degree.

Grigoriev asked if this upper bound of $2k+4$ was tight.  Later work of Grigoriev et al.~\cite{GHP02} showed that the proof technique of \cite{G01} could not show a larger lower bound than $2k+4$.  We show that there actually exist 
Positivstellensatz refutations of~(\ref{eq:knapsack}) of degree $2k+4$.
\begin{theorem}
\label{thm:grig_upper}
Let $0 \le k \le n/2$ be an integer.  For $k < r < k+1$, the system~(\ref{eq:knapsack}) has a Positivstellensatz refutation of degree $2k+4$.  
\end{theorem}

Let $\x=(x_1, \ldots, x_n)$ and $|\x| = \sum_{i=1}^n x_i$.  A key role in the proof will be played by the polynomials 
\[
A_k(\x)= |\x|(|\x|-1)(|\x|-2)\cdots (|\x|-k+1) \enspace.
\]
The function $A_k$ can be computed with $k$ queries by a natural extension of the Sampling \algoref{alg:sample} and thus can be written as a sum-of-squares on the boolean cube of total degree $2k$.  
We go ahead and record this formally in the next lemma. Recall that the $k$th elementary symmetric polynomial is defined as
\[
e_k(x_1,\ldots, x_n) = \sum_{i_1 < i_2 < \cdots < i_k} x_{i_1} x_{i_2} \cdots x_{i_k} \enspace .
\]

\begin{lemma}
\label{lem:key_polys_sos}
There exist polynomials $g_i(\mathbf{x})$ of degree at most $2k-2$ such that
\[
A_k(\mathbf{x}) = \sum_{i=1}^n (x_i^2-x_i) g_i(\mathbf{x}) + (k!) e_k(x_1^2, \ldots, x_n^2) \enspace .
\]
\end{lemma}

We give the proof of this in \apref{app:key_polys_sos}.

\begin{proof}[Proof of \thmref{thm:grig_upper}]
Rearranging \eqnref{eq:ref}, we are looking for functions $g, g_1, \ldots, g_n$ of low degree and a low-degree sum-of-squares $h$ such that 
\[
g(\mathbf{x})(|\x| -r) - 1 = h + \sum_i g_i(\mathbf{x}) (x_i^2-x_i) \enspace .
\]
Notice that, for any $g$, the left-hand side will be negative when $|\x|=r$.
By \lemref{lem:key_polys_sos}, $A_{k+2}$ is of the form of the right-hand side.  Since $A_{k+2}$ has degree $2k+4$, and is also negative when $|\x|=r$, we try to find a polynomial $g(\mathbf{x})$ of degree at most $2k+3$ such that 
\[
g(\mathbf{x})\left(|\x| - r \right) -b= A_{k+2}(\mathbf{x})
\]
for a positive constant~$b$.  Dividing $g$ and $A_{k+2}$ by $b$ will then give us the required solution.  Let $b= -r(r-1)\cdots (r-k)(r-k-1) >0$.  
Then $|\x| -r$ divides $A_{k+2}(\x) + b$ and we can write $A_{k+2}(\x) + b=g(\mathbf{x})\left(|\x| - r \right)$ for some polynomial $g$ of degree~$2k+3$.
\end{proof}

\section{Future work}

We list a few questions for future work:
\begin{itemize}
\item Can we improve the lower bound of \thmref{thm:ell_infty} for small~$\eps$?  To match the upper bound for all~$k$, it would suffice to show that $\sosdeg_\eps(f_1,\ell_\infty)=\Omega(\sqrt{n\log(1/\eps)})$, which is very plausible by analogy with what is known for the $n$-bit OR function.
\item Can we extend our results to all symmetric quadratic functions, or to even larger classes of symmetric functions?
\item Can we find more applications of Blekherman's theorem (\thmref{thm:B}), in complexity theory, in quantum computing, or in optimization? Kurpisz et al.~\cite[Section~5]{KurpiszLeppanenMastrolilli16} used their general reduction to univariate polynomials (already mentioned in~\secref{sec:proofcomplexity}), to show that strengthening the knapsack polytope with Wolsey's ``Knapsack Covering Inequalities''  and applying nearly $\log n$ rounds of the Lasserre hierarchy does not produce an SDP with integrality gap below~$2-o(1)$ (which is the integrality gap of the natural LP relaxation). Similar results may be obtainable using Blekherman's theorem.
\end{itemize}

\subsection*{Acknowledgments}
We thank Mario Szegedy for early discussions on lower bounding the sos degree of~$f_1$, Srinivasan Arunachalam for 
many useful discussions and comments, Greg Blekherman for discussions about~\cite{B15}, and Adam Kurpisz for alerting us to~\cite{KurpiszLeppanenMastrolilli16}.   We thank the authors of the software SDPT3~\cite{TTT99} and YALMIP~\cite{Lof04}, which we used to compute the error of sos polynomials in approximating $f_k$ on small examples.  TL would like to thank Savaroo Ranch for their 
hospitality and BBQ during the course of this research.


\newcommand{\etalchar}[1]{$^{#1}$}

\appendix

\section{Proof of \lemref{lem:key_polys_sos}}
\label{app:key_polys_sos}
Recall that 
\[
A_k(x_1, \ldots, x_n) = |\x|(|\x|-1) \cdots (|\x|-k+1) \enspace .
\]

We first prove two claims.
\begin{claim}
\label{claim1}
There exist polynomials $g_i(\mathbf{x})$ of degree at most $k-1$ such that
\[
A_k(x_1, \ldots, x_n) = \sum_{i=1}^n (x_i^2-x_i) g_i(\x) + (k!) e_k(\x) \enspace .
\]
\end{claim}

\begin{proof}
We prove the claim by induction on~$k$.  When $k=1$ then $A_1 = e_1$, and the claim follows by setting all $g_i$ to be 
$0$.  

Now suppose the claim is true up to $k$.  Then, using the induction hypothesis to rewrite $A_k$, 
\begin{align*}
A_{k+1}(\x) = A_k(\x) \cdot (e_1(\x) - k) &= \left( \sum_{i=1}^n (x_i^2 - x_i) g_i(\x) + k! e_k(\x) \right) 
(e_1(\x) - k) \\
&= \sum_{i=1}^n (x_i^2-x_i)h_i(\x) + k! e_k(\x)( e_1(\x) -k) \enspace ,
\end{align*}
where each $h_i(\x)=g_i(\x) \cdot ( e_1(\x) -k)$ is of degree at most $k$.  We now focus on
\[
e_k(\x)( e_1(\x) -k) = \sum_{\substack{S \subseteq [n] \\ |S|=k}} \prod_{i \in S} x_i \cdot ( e_1(\x) -k) \enspace .
\]
A term in this sum corresponding to the subset $S$ can be rewritten as
\begin{align*}
\prod_{i \in S} x_i \left(\sum_{i \in S} x_i + \sum_{i \not \in S} x_i -k \right) =
\sum_{i \in S} (x_i^2 - x_i) \prod_{j \in S, j \ne i} x_j + \sum_{i \not \in S} x_i \prod_{j \in S} x_j
\end{align*}
Summing over all terms of this form gives $(k+1)! e_{k+1}(\x) + \sum_i (x_i^2-x_i) \cdot f_i(\x)$, where $f_i(\x)$ is of 
degree at most $k-1$, proving the claim.
\end{proof}

To complete the proof, we now need to show that $e_d(\x)$ is a sum of squares of total degree $2d$ modulo the ideal 
$\langle x_1^2 -x_1, \ldots, x_n^2-x_n\rangle$.  To do this, it suffices to show the same for $\prod_{i=1}^d x_i$, 
which we do in the next claim.
\begin{claim}
\label{claim:sos}
Fix a natural number $d$.  Then there are polynomials $g_i \in \mathbb{R}[x_1, \ldots, x_d]$ for $i =1, \ldots, d$ such 
that
\[
\prod_{i=1}^d x_i^2 - \prod_{i=1}^d x_i = \sum_{i=1}^d (x_i^2 - x_i) g_i(x_1, \ldots, x_d) \enspace,
\]
and each $g_i$ is of degree at most $2d-2$.
\end{claim}

\begin{proof}
We write $x_1^2 x_2^2 \cdots x_d^2-x_1 x_2 \cdots x_d$ as a telescoping sum.  We use the convention that the product 
over the empty set is $1$.
\begin{align*}
\prod_{i=1}^d x_i^2 - \prod_{i=1}^d x_i &= \sum_{j=1}^d \left(\prod_{i<j} x_i \prod_{i \ge j} x_i^2 - 
\prod_{i\le j} x_i \prod_{i > j} x_i^2 \right) \\
&= \sum_{j=1}^d (x_j^2 - x_j) \prod_{i<j} x_i \prod_{i > j} x_i^2 
\end{align*}
This is of the desired form, and it can be seen that each multiplier of $x_j^2-x_j$ is of degree at most $2d-2$.  
\end{proof}

We put these claims together to prove \lemref{lem:key_polys_sos}, which we restate here.
\begin{lemma}
There exist polynomials $g_i(\mathbf{x})$ of degree at most $2k-2$ such that
\[
A_k(\mathbf{x}) = \sum_{i=1}^n (x_i^2-x_i) g_i(\mathbf{x}) + (k!) e_k(x_1^2, \ldots, x_n^2) \enspace .
\]
\end{lemma}

\begin{proof}
By \claimref{claim1} we can write $A_k(x_1, \ldots, x_n) = \sum_{i=1}^n (x_i^2-x_i) g_i(\x) + (k!) e_k(\x)$ where each 
$g_i(\x)$ is of degree at most $k-1$.  Now by \claimref{claim:sos}
\[
e_k(\x) = e_k(x_1^2, \ldots, x_n^2) + \sum_{i=1}^n (x_i^2-x_i) \cdot f_i(\x) \enspace ,
\]
where each $f_i(\x)$ is of degree at most $2k-2$.  This proves the lemma.
\end{proof}

\section{Blekherman's theorem} \label{ap:B}  

Blekherman and Riener~\cite{BR12} made a general study of the relationship between symmetric nonnegative forms and symmetric sums of squares.  Subsequently, Blekherman~\cite{B15} considered the special case of polynomials that are nonnegative on the hypercube, and gave a very useful decomposition of such polynomials.  We include a proof here of a special case of his theorem. 

The technique used for our proof is a novel decomposition of functions on the hypercube using the kernels of certain differential 
operators. A similar decomposition was independently discovered by Filmus and Mossel~\cite{FM16} who use it to prove an invariance principle for low-degree functions on slices of the boolean hypercube. 

Let $[n]$ denote the set of integers $\{1,2,\ldots,n\}$. The ideal $\mathcal{I}:= \langle x_{i}^{2} - x_{i}:i\in[n] \rangle$ consists of polynomials that are identically zero on the hypercube $\mathcal{H}=\{0,1\}^{n}$. Let $L_{t}=\R[x]_{t}/\mathcal{I}$ be the  space of degree-$t$ homogeneous multilinear polynomials on $n$ variables. The $\binom{n}{t}$ monomials $x^{S}$ where 
$S \subseteq [n], |S| = t$, form a basis for $L_{t}$. The correspondence between set $S$ and monomial $x^{S}$ can be used to 
map degree-$t$ polynomials to linear combinations of $t$-subsets of $[n]$. 
Here polynomial $p(x)$ corresponds to the $\binom{n}{t}$-dimensional vector of the coefficients of its monomials, say in lexicographic order.

Let $M_{t}=\R[x]_{\leq t}/\mathcal{I}$ denote the space of $n$-variate polynomials of degree at most $t$ on the hypercube. Given $x \in \R^{n}$, the sum $\sum_{i \in [n]} x_{i}$ is denoted by $|x|$.

\subsection{Kernels of the operators $W_{t}$} \label{sec1} 

For $t\geq 0$, define the linear operator $W_{t}$ that acts by summing over the partial derivatives of a degree-$t$ polynomial, 
\all{ 
W_{t} p(x) = \en{ \sum_{i\in [n]} \part{}{x_{i} } } p(x) \enspace .
} {one} 
For $t\geq 1$ the operator $W_{t}:L_{t} \to L_{t-1}$ is represented by a matrix with rows and columns indexed by $S, T \subset [n]$ with $|S|=t-1$ and $|T|=t$ respectively, and entry $(W_{t})_{S, T}=1$ if $S \subset T$ and $0$ otherwise. 
The adjoint operator $W_{t}^{T}: L_{t-1} \to L_{t}$ acts as multiplication by $|x|-t+1$ on each degree-$(t-1)$ monomial (and by linear extension on all of $L_{t-1}$),
\all{ 
W_{t}^{T} x^{S}  = \sum_{i\not \in S} x^{S \cup \{i\}} = x^{S} ( |x| - t+1) \enspace .
} {two} 
Note that we used the hypercube constraints $x_{i}^{2} = x_{i}$ to derive the second equality in \eqnref{two}. 
Our goal in this section is to bound the dimension of $\Ker(W_{t})$ and find an explicit basis for these spaces. 

We relate $\Ker(W_{t})$ to the eigenspaces of the Johnson graphs. The Johnson graph $J(n,t)$ has $\binom{n}{t}$ vertices corresponding to the $t$-subsets $S \subset [n], |S|=t$, with subsets $S, T$ connected by an edge if and only if $|S \cap T|= t-1$. The adjacency matrix of $J(n,t)$ is denoted by 
$A_{J}(n, t)$. The following lemma computes the spectrum of $A_{J}(n,t)$; it can be found in Godsil's notes~\cite{G10}, 
but we include a proof here as these notes are no longer online.

\begin{theorem} 
The eigenvalues of $A_{J}(n, t)$ are $t(n-t) - i(n+1-i)$ with multiplicity $\binom{n}{i} - \binom{n}{i-1}$ for $i=\{0,1,\ldots,t\}$
and $t\leq (n+1)/2$. 
\end{theorem} 

\begin{proof} 
We proceed by induction on $n$ and~$t$. For the base case, note that the Johnson graph $J(n,1)$ is the complete graph on $n$ vertices. The corresponding adjacency matrix $A_{J}(n,1)$ has eigenvalue $-1$ with multiplicity $(n-1)$ and $(n-1)$ with multiplicity $1$, thus the theorem is true for $n=1$. 

We obtain the spectrum of $A_{J}(n,t)$ in terms of the spectrum of $A_{J}(n,t-1)$. Computing the entries of $W_{t}^{T} W_{t}$ and 
$W_{t} W_{t}^{T}$ it follows that, 
\al{ 
W_{t}^{T} W_{t}&= tI + A_{J}(n, t) \nl 
W_{t} W_{t}^{T}&=  (n-t+1)I   + A_{J}(n,t-1) \enspace .
} 
 The non-zero eigenspaces of $W_{t}^{T} W_{t}$ correspond to those of $W_{t}W_{t}^{T}$, so if $v$ is an eigenvector for $A_{J}(n,t-1)$ with eigenvalue $\lambda_{i}$ then $W_{t}^{T}v$ is an eigenvector for $A_{J}(n,t)$ with eigenvalue $\lambda_{i} + n-2t+1$. By the induction hypothesis, $(t-1)(n-t+1)- i(n+1-i)$ 
is an eigenvalue for $A_{J}(n,t-1)$ with multiplicity $\binom{n}{i} - \binom{n}{i-1}$ for $i \in [t-1]$. Adding $n-2t+1$, it follows that $t(n-t) - i(n+1-i)$ 
is an eigenvalue for $A_{J}(n,t)$ with the same multiplicity. 

The induction hypothesis also implies that $W_{t} W_{t}^{T} = (n-t+1)I   + A_{J}(n,t-1)$ has rank $\binom{n}{t-1}$ as it is positive semidefinite and the smallest eigenvalue is $n-2t+2>0$. Hence the $\binom{n}{t}$-dimensional matrix $W_{t}^{T} W_{t}$ has rank $\binom{n}{t-1}$, so its kernel has dimension $\binom{n}{t} - \binom{n}{t-1}$.  This implies that $A_{J}(n, t)$ has an eigenspace of dimension $\binom{n}{t} - \binom{n}{t-1}$ with eigenvalue $-t= t(n-t) - t(n+1-t)$. 
\end{proof} 
\noindent The following corollary computes the dimension of $\Ker(W_{t})$. 

\begin{lemma} \label{c0} 
$\Dim (\Ker(W_{t}))=\binom{n}{t} - \binom{n}{t-1}$ for $t\leq (n+1)/2$.
\end{lemma} 

\begin{proof} 
$\Dim(\Ker(W_{t}))=\binom{n}{t}  - \rank(W_{t}^{T})$, and from the above proof it follows that $\rank(W_{t}^{T})=\rank(W_{t}W_{t}^{T})= \binom{n}{t-1}$ for 
$t\leq (n+1)/2$. 
\end{proof} 

\noindent We next compute an explicit basis for $\Ker(W_{t})$, viewed as a subspace of $L_{t}$. We recall the notion of 
a standard Young tableau of shape $(n-t, t)$ to describe the basis. 

\begin{defn} 
A standard Young tableau $\mathcal{U}$ of shape $(n-t, t)$ is an arrangement of $[n]$ in an array with two rows of 
size $n-t$ and $t$ respectively, such that each row and column is sorted in ascending order. 
\end{defn} 

\noindent The basis for $\Ker(W_{t})$ described by the following theorem will be used for computations in the following sections. Note that polynomials $p_{\mathcal{U}}$ in this basis evaluate to~0 for all $x\in \{0,1\}^{n}$ with $|x| \in \{ 0, 1, \ldots, t-1 \} \cup \{ n, n-1, \ldots, n-t+1 \}$.

\begin{theorem} \label{Span} 
For $t \leq n/2$ and $\mathcal{A}= (a(1), a(2), \ldots, a(2t))$ an array of distinct elements $a(i) \in [n]$, define the polynomial $p_{\mathcal{A}}(x):=\prod_{i \in [t]} (x_{a(2i-1)} - x_{a(2i)})$. 

The polynomials $p_{\mathcal{U}}(x)$, where $(u(2i-i), u(2i))$ for $i\in [t]$ are the entries of the $i$-th column of a standard $(n-t, t)$ Young tableau~$\mathcal{U}$, form a basis for $\Ker(W_{t})$.   
\end{theorem} 

\begin{proof} 
We first show that for all $|\mathcal{A}|=2t$, the degree-$t$ polynomial $p_{\mathcal{A}}(x)$ belongs to the kernel of~$W_{t}$. Computing the 
partial derivatives of $p_{\mathcal{A}}(x)$, 
\al{ 
\part{}{x_{j}} p_{\mathcal{A}}(x) = 
\begin{cases} 
p_{\mathcal{A}}(x)/(x_{a(2i-1)} - x_{a(2i)}) \text{ if $j=a(2i-1)$} \\
-p_{\mathcal{A}}(x)/(x_{a(2i-1)} - x_{a(2i)}) \text{ if $j=a(2i)$} \\
0 \text{ otherwise} 
\end{cases}  
} 
Summing over the partial derivatives and using \eqnref{one} it follows that $W_{t} p_{\mathcal{A}}(x) =0$. 

The set of polynomials $\{p_{\mathcal{A}}(x): |\mathcal{A}|=2t\}$ is not linearly independent. The straightening algorithm for Young tableaux (see for example Section~10.5 in~\cite{CST08}) shows that the polynomials $p_{\mathcal{U}}(x)$ where $(u(2i-i), u(2i))$ for 
$i\in [t]$ are entries of the $i$-th column of a standard $(n-t, t)$ Young tableau $\mathcal{U}$ form a basis for $\Span\{p_{\mathcal{A}}(x):|\mathcal{A}|=2t\}$. A simple counting argument or the hook length formula \cite{CST08} shows that the number of such $~\mathcal{U}$ is $\binom{n}{t} - \binom{n}{t-1}$. These $p_{\mathcal{U}}$ together thus span a space of dimension $\binom{n}{t} - \binom{n}{t-1}$, which is $\Dim(\Ker(W_{t}))$ by \lemref{c0}. Hence the $p_{\mathcal{U}}$ form a basis for~$\Ker(W_{t})$. 
\end{proof}

\subsection{Polynomial decompositions} 
The action of the operators $W_{t}$ yields the decomposition $L_{t} = \Ker(W_{t}) \oplus \Img(W_{t}^{T})$. Applying this decomposition iteratively we obtain the following theorem, 
\begin{theorem} \label{dec1} 
A polynomial $p(x) \in L_{t}$ can be decomposed as 
\all{ 
p(x) &= p_t(x) + (|x|-t+1) p_{t-1}(x) + 
\cdots + (|x|-t+1)\cdots (|x|-1)|x| p_0(x) \nl 
&= p_t(x) + \sum_{i=1}^t p_{t-i}(x) \prod_{j=1}^{i} (|x|-t+j)
} {six} 
where $p_{t-i}(x) \in \Ker(W_{t-i})$. 
\end{theorem} 

\begin{proof} 
We proceed by induction on~$t$.  For the base case $t=0$, observe that a degree-$0$ polynomial belongs to $\Ker(W_{0})$. 
For the inductive step, a polynomial $p(x) \in L_{t}$ can be written as $p_{t}(x) + q(x)$ where $p_{t}(x) \in \Ker(W_{t})$ and $q(x) \in \text{Im}(W_{t}^{T})$.  
The action of $W_{t}^{T}$ on polynomials in $L_{t-1}$ is described by \eqnref{two}: for all $g(x) \in L_{t-1}$ we have
\al{ 
W_{t}^{T} g(x) =( |x| - t+1) g(x) \enspace .
} 
As $q(x) \in \Img(W_{t}^{T})$, it can be factored as $q(x)=(|x|- t+1) h(x)$ where $h(x)\in L_{t-1}$.  The result follows using the induction hypothesis for $g(x)$.  
\end{proof} 

Applying the above theorem to the subspaces $L_{j}$ ($j\in\{0,\ldots,t\}$) that are contained in $M_{t}$ and collecting the terms corresponding to $\Ker(W_{j})$, we obtain the following decomposition for polynomials in $M_{t}$. 

\begin{corollary} \label{cor5} 
A polynomial $p(x) \in M_{t}$ can be decomposed as $p(x)= \sum_{j=0}^t q_{j}(x)$, where
\all{ 
q_{j}(x) = \sum_{0\leq i\leq t-j} |x|^{i} p_{ij} (x)    
} {cor6} 
such that each $p_{ij}(x) \in \Ker(W_{j})$. 
\end{corollary}

\begin{proof} 
A polynomial $p(x) \in M_{t}$ can be written as $p(x) = \sum_{i=0}^t p_{i}(x)$ where $p_{i}(x) \in L_{i}$ is the homogeneous degree-$i$ component of $p$. Applying \thmref{dec1} to each $p_{i}(x)$ and collecting all the terms over the \eqnsref{six} with prefactors in $\Ker(W_{j})$, we obtain a decomposition 
$p(x) = \sum_{j \in [t]} q_{j}^{\prime}(x)$ such that 
\al{ 
q_{j}^{\prime}(x) &= p_{jj}^{\prime}(x) + p_{j+1, j}^{\prime} (x) (|x| -j) + p_{j+2, j}^{\prime}(x) (|x|-j)(|x|-j+1)+ \cdots \nl 
&\cdots + p_{t, j}^{\prime} (x) (|x|-j)(|x|-j+1)\cdots(|x| - t+1) \enspace .
} 
Note that the indices in the above equation increase, because $\Ker(W_{j})$ occurs in the decompositions of $p_{i}(x)$ for $i\geq j$. 
Let $p_{ij}(x)$ be the coefficient of $|x|^{i}$ for $0\leq i \leq t-j$ in the above expression. This $p_{ij}(x)$ is a linear combination of polynomials $p^{\prime}_{ij}(x) \in \Ker(W_{j})$ 
and therefore also lies in $\Ker(W_{j})$. The decomposition in \eqnref{cor6} follows. 
\end{proof} 

\subsection{Symmetrization and Blekherman's theorem} 
The symmetric group $S_{n}$ acts on the polynomial ring $M_{n}$ by permuting the indices of the monomials. The 
subspace of symmetric polynomials in $M_{n}$ that are invariant under the action of $S_{n}$ is denoted by $\Lambda_{n}$. 
The operator $\Sym: M_{n} \to \Lambda_{n}$ maps a polynomial to its symmetrization, 

\all{ 
\Sym(p)(x) := \frac{1}{n!} \sum_{\sigma \in S_{n}} p(\sigma x)  \enspace .
} {sym}
The symmetrization of degree $k$ monomials evaluates to a univariate polynomial in $|x|$ over $M_{n}$. 
\begin{lemma} \label{2syms} 
Let $m_{k}(x)=x_{1}x_{2}\cdots x_{k}$ be a degree-$k$ monomial, then the following identity is true in the ring $M_{n}$, 
\all{ 
\Sym(m_{k})(x) = \frac{|x|(|x|-1)\cdots (|x|-k+1) }{ n(n-1) \cdots (n-k+1)} \enspace .
} {eleven}
\end{lemma} 
\begin{proof} 
We proceed by induction on $k$, for $k=1$ the result is clearly true. Let $x^{S}$ be an arbitrary degree $k$ monomial. There are $k!(n-k)!$ permutations $\sigma \in S_{n}$ such that 
$\sigma (x_1 x_2 \cdots x_k)= x^{S}$, thus $\Sym(m_{k})(x)$ evaluates to
\all{ 
\Sym(m_{k})(x) &= \frac{1}{ \binom{n}{k} } \sum_{|S|=k} x^{S} \enspace .
} {sympk} 
In order to to express $\Sym(m_{k})(x)$ in terms of $\Sym(m_{k-1})(x)$, we write the above equation in terms of the operator $W_{k}^{t}$ from \secref{sec1}, 
\al{ 
\Sym(m_{k})(x) &=  \frac{1}{ \binom{n}{k} } W_{k}^{T} \en{  \frac{1}{k} \sum_{|U|=k-1} x^{U} } =  \frac{(|x| - k+1)}{(n-k+1)} \Sym(m_{k-1})(x)  \enspace .
} 
The second equality follows from \eqnref{two} and the expression for $\Sym(m_{k-1})$ in \eqnref{sympk}. The result follows from the induction hypothesis. 
\end{proof} 
\noindent The above lemma shows that $\Sym(p)$ for polynomials $p \in M_{n}$ can be viewed as a univariate polynomial in 
$|x|$ by extending the mapping given by \lemref{2syms} to all $p\in M_{n}$. We denote the univariate polynomial 
thus obtained by $\Sym^{uni}(p)$ to disambiguate from the multivariate polynomial in \eqnref{sym}. 

We can define an inner product on polynomials $p,q \in L_{t}$ 
by treating them as vectors of coefficients: if $p(x)= \sum_{|S|=t} p_{S}x^{S}$ and $q(x)= \sum_{|S|=t} q_{S}x^{S}$ then
\al{ 
\braket{p} {q} := \sum_{S\subseteq [n], |S|= t} p_{S} q_{S} \enspace .
} 
The symmetrization of the product of polynomials in $\Ker(W_{t})$ can be expressed in terms of this inner product.

\begin{lemma} \label{lem2} 
If $p, q \in \Ker(W_{t})$ for some $t\leq n/2$, then: 
\all{ 
\Sym(pq)(x)= \braket{p}{q} \frac{(n-2t)!}{n! } \prod_{0\leq i <t}  (|x|-i)(n-|x|-i)\enspace .
} {poly} 
\end{lemma} 

\begin{proof} 
\thmref{Span} shows that $\Ker(W_{t})$ has a basis consisting of polynomials $p_{\mathcal{U}}(x)$ such that $p_{\mathcal{U}}(x)=0$ for all $x\in \{0,1\}^{n}$ with $|x| \in \{ 0, 1, \ldots, t-1 \} \cup \{ n, n-1, \ldots, n-t+1 \}$. Consider such an~$x$. Evaluating $\Sym(pq)$ at $x$ using \eqnref{sym} by expanding $p$ and $q$ in the basis given by \thmref{Span}, it follows that $\Sym^{uni}(pq)(\alpha)=0$ for all 
$\alpha \in \{ 0, 1, \ldots, t-1 \} \cup \{ n, n-1, \ldots, n-t+1 \}$. \lemref{2syms} shows that $\Sym(p q)(x)$ is a univariate polynomial $\Sym^{uni}$ in $|x|$ of degree at most $2t$, hence  
\al{
\Sym(p q)(x)=\lambda \prod_{0\leq i <t}  (|x|-i)(n-|x|-i)\enspace .
}
for some $\lambda \in \R$. Below we determine $\lambda$ by evaluating 
$\Sym(pq)$ for $x \in \{0,1\}^{n}$ such that $|x|=t$. 

We compute $\Sym(pq)(x)$ by evaluating the sum $\sum_{\sigma \in S_{n}} p(\sigma x) q(\sigma x)$ in \eqnref{sym}. As $p, q$ are homogeneous degree-$t$ polynomials, for each $x$ with $|x|=t$ there is a unique $S \subset [n]$, $|S|=t$, such that $p(x)=p_{S}$ and $q(x)=q_{S}$. In other words, $x$ sets exactly one degree-$t$ monomial $x^S$ to~1 and all others to~0. There are $t!(n-t)!$ different $\sigma\in S_{n}$ such that $\sigma(x)$ sets the same monomial to~1. The symmetrization $\Sym(pq)(x)$ therefore evaluates to
\al{ 
\Sym( pq )(x)= \frac{1}{n!} \sum_{\sigma \in S_{n}} p(\sigma x) q(\sigma x)= \frac{t!(n-t)!}{n!} \sum_{|S|=t} p_{S} q_{S} \enspace .
}  
$\Sym(pq)(x)$ also evaluates to $\lambda  \prod_{0\leq i <t}  (t-i)(n-t-i)$. Equating the two expressions we have: 
\al{ 
 \lambda t! \prod_{0\leq i< t} (n-t-i) = \frac{ \braket{p}{q} t! (n-t)! }{ n! }  
}
which implies $\lambda =  \frac{ \braket{p}{q} (n-2t)! }{ n! } $, and the theorem follows. 
\end{proof} 

\noindent We next show that the symmetrization of the product of polynomials $p \in \Ker(W_{t}), q \in \Ker(W_{t^{\prime}})$ evaluates to $0$ if $t\neq t^{\prime}$. The following lemma is used for the proof in \lemref{lem3}. 

\begin{lemma} \label{lem1} 
If $p(x)= \prod_{i \in [k]} (x_{i} - x_{i+1})q(x)$ for some odd~$k$, and $q(x)$ is a polynomial that does not depend on variables $x_{1},\ldots,x_{k+1}$, then $\Sym(p)=0$. 
\end{lemma}
 
\begin{proof} 
It suffices to show that $\Sym(p)(x)=0$ for all $x\in\{0,1\}^n$, because a multilinear polynomial that is~0 on the hypercube is identically equal to~0.
Define the involution $\sigma \to \overline{\sigma}$ on $S_{n}$ by setting $\overline{\sigma}(i)= \sigma(i+1)$ if $i\in [k+1]$ is odd, 
$\overline{\sigma}(i)= \sigma(i-1)$ if $i \in [k+1]$ is even, and $\overline{\sigma}(i)= \sigma(i)$ for $i > k+1$. It follows that $\overline{\sigma}$ is an involution as $k+1$ is even and it acts by swapping the pairs $(\sigma(2i-1), \sigma(2i))$ for $i\in [(k+1)/2]$. This involution partitions $S_n$ into pairs $(\sigma,\overline{\sigma})$, and hence
\al{ 
\Sym(p) (x) = \frac{1}{n!} \sum_{(\sigma, \overline{\sigma}) } (p(\sigma x) + p(\overline{\sigma} x))  = 0 .
}  
The second equality follows as $p(\overline{\sigma} x) = -p(\sigma x)$ for all $x \in \{0,1\}^{n}$ and $\sigma \in S_{n}$. 
\end{proof}

\begin{lemma} \label{lem3} 
If $p \in \Ker(W_{t})$ and $q \in \Ker(W_{t^{\prime}})$ for $n/2 \geq t > t^{\prime}$, then $\Sym(pq)=0$. 
\end{lemma} 

\begin{proof} 
It suffices to prove the statement for polynomials $p=p_{\mathcal{U}}$ and $q=q_{\mathcal{V}}$ belonging to the bases for $\Ker(W_{t})$ and $\Ker(W_{t^{\prime}})$ constructed in \thmref{Span}. The arrays $\mathcal{U}, \mathcal{V}$ define matchings $M(\mathcal{U})= \bigcup_{i\in [t]} (u(2i-1), u(2i))$ and $M(\mathcal{V}) = \bigcup_{i\in [t^{\prime}]} (v(2i-1), v(2i))$ on $[n]$ of size $t$ and $t^{\prime}$ respectively. The product $p_{\mathcal{U}}q_{\mathcal{V}}= \prod_{(a, b) \in M(\mathcal{U}) \cup M(\mathcal{V})} (x_{a} - x_{b})$. If $M(\mathcal{U}) \cup M(\mathcal{V})$ contains an odd-length path as an induced subgraph, then we can use \lemref{lem1} to conclude that $\Sym(p_{\mathcal{U}}q_{\mathcal{V}})=0$. 

It suffices to show that the union of two matchings of different sizes contains an odd-length path as an induced subgraph.  
The connected components of a union of two distinct matchings on $[n]$ either form even-length cycles or paths. Color the edges in $M(\mathcal{U})$ red and the edges in $M(\mathcal{V})$ blue. The number of red edges $t$ is greater than blue edges $t^{\prime}$, so there must be at least one connected component that is an odd-length path, as even-length paths and cycles have an equal number of red and blue edges. 
\end{proof} 

The preceding lemmas allow us to give a proof of Blekherman's result~\cite{B15} on the symmetrization of sum-of-squares polynomials on the hypercube. 

\begin{theorem}[Blekherman] \label{sos} 
The symmetrization of the square of polynomial $p \in M_{t}$ for $t \leq n/2$ can be decomposed as
\all{ 
\Sym(p^{2})(x) = \sum_{j=0}^t p_{t-j}(|x|)  \en{ \prod_{0\leq i <j} (|x| - i) (n-|x| -i) } 
} {onenine} 
where $p_{t-j}$ is a univariate polynomial that is the sum of squares of polynomials of degree at most $t-j$. 
\end{theorem} 

\begin{proof} 
Consider the representation of the polynomial $p(x)= \sum_{j=0}^t q_{j}(x)$ given by \corref{cor5}, 
\all{ 
q_{j}(x) = \sum_{0 \leq k \leq t-j} |x|^{k} p_{kj}(x) 
} {cor7} 
where the polynomials $p_{kj}(x) \in \Ker(W_{j})$. \lemref{lem3} shows that $\Sym(p_{kj} p_{k^{\prime}j^{\prime}})=0$ if $j\neq j^{\prime}$, 
hence $\Sym(p^{2})$ can be decomposed as 
\al{ 
\Sym(p^{2}) = \sum_{j=0}^t \Sym \en{ q_{j}^{2}   } \enspace .
} 
Expanding the term $\Sym \en{ q_{j}^{2} }$ using \lemref{lem2}, we have
\al{ 
\sum_{0\leq k, l \leq t-j} \Sym( |x|^{k+l} p_{kj} p_{lj}  )   &=  c \en{ \prod_{0\leq i <j} (|x| - i) (n-|x| -i) }  \sum_{0\leq k, l \leq t-j} \braket{ p_{kj} }{ p_{lj} }|x|^{k+l}  \nl 
&=  c \en{ \prod_{0\leq i <j} (|x| - i) (n-|x| -i) }  \mathbf{x}^{T} P \mathbf{x}  
} 
where $c$ is a constant independent of $|x|$, $P \in \R^{(t-j+1)\times (t-j+1)}$ is the matrix with entries $P_{kl} = \braket{ p_{kj} }{ p_{lj} }$, and $\mathbf{x} \in \R^{t-j+1}$ is the vector with entries $(1, |x|, |x|^{2}, \ldots, |x|^{t-j})$. The matrix $P$ is positive semidefinite, hence the polynomial $p_{t-j}(|x|)$ corresponding to the quadratic form $\mathbf{x}^{t} P \mathbf{x}$ is a sum of squares of polynomials in $|x|$ of degree at most $t-j$. The theorem follows. 
\end{proof} 

Note that the proof is constructive, as it provides a way to compute the terms in the decomposition by projecting onto the eigenspaces $W_{t}$ of the Johnson scheme. For example, the first term $p_{t}(|x|)$ in \eqref{onenine} is in fact $\Sym^{uni}(p)^{2}$ as the $\Sym$ operator maps $\Ker(W_{j})$ to $0$ for all $j>0$. 

\begin{corollary} 
The polynomial $p_{t}(|x|)$ in \thmref{sos} is $\Sym^{uni}(p)^{2}$. 
\end{corollary} 

\noindent A symmetric function $f$ that is the sum of squares of polynomials of degree $d\leq n/2$ is a sum of terms 
$\Sym(p^{2})$ for $n$-variate polynomials $p$ of degree $d\leq n/2$. Applying \thmref{sos} for $t=d$ we 
obtain Blekherman's result as stated in \thmref{thm:B}.  

Note that \thmref{sos} 
applies to the setting where $deg(p(x))\leq n/2$, this suffices for our applications. Blekherman's theorem in ~\cite{B15} is valid for all degrees modulo the ideal 
$I= \langle \prod_{0\leq i\leq n} (|x|-i) \rangle$.

\section{Grigoriev's knapsack lower bound} \label{ap:C} 
We now see how Blekherman's theorem can be easily used to reprove Grigoriev's lower bound on 
the degree of Positivstellensatz refutations of knapsack (\thmref{thm:grig}).  A Positivstellensatz refutation of 
the knapsack system of equations \eqref{eq:knapsack_intro} with parameter $r$ consists of polynomials 
$g, g_1, \ldots, g_n$ and a sos polynomial $h$ such that 
\begin{equation}
\label{eq:app_appendix}
g(x) \cdot \left(\sum_{i=1}^n x_i - r \right) + \sum_{i=1}^n g_i(x) \cdot (x_i^2-x_i) = 1 + h(x) \enspace.
\end{equation}

\begin{theorem}[Grigoriev \cite{G01}]
\label{thm:G}
Let $0 \le k \le (n-3)/2$ be an integer.  If $k < r < n-k$, then any Positivstellensatz refutation of the knapsack 
system of equations with parameter $r$, as in~\eqnref{eq:app_appendix}, has degree at least $2k+4$.
\end{theorem}

\begin{proof}
Grigoriev constructs a functional $\G_r : \R[x_1, \ldots, x_n] \rightarrow \R$ such that: when $\G_r$ is applied to 
the left-hand side of~\eqnref{eq:app_appendix} it evaluates to $0$, provided that the total degree of the left-hand side is 
at most $n$; and when $\G_r$ is applied to the right-hand side of~\eqnref{eq:app_appendix} it is at least $1$, provided the total degree of the right-hand side is at most $2k+2$.  This leads to a contradiction, hence constructing such a 
functional $\G_r$ suffices to prove that a Positivstellensatz refutation must have degree at least $\min\{n, 2k+4\}$.  The theorem then 
follows, with the additional observation that if $\min\{n, 2k+4\}=n$ is odd, then we can actually obtain a lower bound of 
$n+1$ (since any sum-of-squares polynomial must have even degree).  

The functional $\G_r$ is first defined on the quotient ring 
$\mathcal{A}=\R[x_1, \ldots, x_n] / \langle x_1^2-x_1, \ldots, x_n^2-x_n\rangle$.  For $p \in \mathcal{A}$ define 
$$
\G_r(p) = \Sym^{uni}(p)(r).
$$
In other words, $\G_r$ looks at the univariate polynomial formed from the 
symmetrization of $p$ over the symmetric group, and evaluates it at the point $r$.  Explicitly, for a monomial 
$x_S = \prod_{i \in S} x_i$ with $|S|=t$ we see by~\lemref{2syms} that $\G_r(x_S)=B_{t}$ where
\begin{equation}
\label{eq:B}
B_{t} = \frac{r(r-1) \cdots (r-t+1)}{n(n-1)\cdots (n-t+1)} \enspace .
\end{equation}

For $p \in \R[x_1, \ldots, x_n]$ let $\bar p$ be its canonical multilinear representative in $\mathcal{A}$.  The 
definition of the functional $\G_r$ is extended from $\mathcal{A}$ to the polynomial ring by letting 
$\G_r(p):=\G_r(\bar p)$ for $p \in \R[x_1, \ldots, x_n]$.

Grigoriev's theorem now follows from the following four observations about $\G_r$:
\begin{enumerate}
\item $\G_r\left(g(x) \cdot (\sum_i x_i -r)\right)=0$ for all polynomials $g$ with $\deg(g) < n$.  It suffices to show this 
for $g(x)=x_S = \prod_{i \in S} x_i$ for some $S \subsetneq [n]$ with $|S|=t < n$.  In this case, by~\eqnref{eq:B},
$\G_r(x_S(\sum_i x_i -r))= (n-t) B_{t+1} + (t-r)B_t=0$.

\item $\G_r \en{ g_{i}(x) (x_{i}^{2} - x_{i}) }=0$ for all polynomials $g_{i}$. This is because the canonical multilinear representative of $g_{i}(x) (x_{i}^{2} - x_{i})$ in the quotient ring $\mathcal{A}$ is the constant-0 polynomial, and $\Sym^{uni}(0) (r)=0$. 

\item $\G_r(1) =1$ for all values of $r$.  The symmetrization of the constant-1 polynomial is itself, and the constant-1 polynomial always evaluates to 1.

\item $\G_r \en{ p^2(x) }\geq 0$ if $p$ is a polynomial of degree at most $k+1$.  By Blekherman's 
Theorem~\ref{thm:B}, if $p \in \mathcal{A}$ and $d=\deg(p)$ then
\al{ 
\Sym^{uni}(p^{2})(x) = q_{d}(x) + x(n-x) q_{d-1}(x) + x(x-1)(n-x)(n-1-x) q_{d-2}(x) + \cdots \nl 
+ x(x-1)\cdots(x-d+1)(n-x)(n-1-x)\cdots(n-d+1-x) q_{0} (r) \enspace .
} 
It follows that $\Sym^{uni}(p^{2})(x) \ge 0$ for $x \in [d-1, n-d+1]$.  Thus if $k < r < n-k$, then $\G_r(p^2) \ge 0$ for 
any $p$ of degree $\le k+1$. By linearity this extends to any $h$ that is a sum of squares of polynomials of degree $\le k+1$.
\end{enumerate}
The first two observations imply that the left-hand side of~\eqnref{eq:app_appendix} evaluates to $0$ under $\G_r$ (provided the total degree of the left-hand side is at most $n$), while the last two observations imply that the right-hand side evaluates to at least~1 (provided the total degree on the right is at most $2k+2$).
\end{proof}

\section{Application of Grigoriev's bound to $\ell_\infty$-error sos degree}
\label{ap:D}
Let $n$ be odd and let $f=f_{\floor{n/2}}$, that is $f(x) = (|x|-n/2)^2-1/4$.  
Our \thmref{thm:ell_infty} gives 
that any sos polynomial approximating~$f$ with $\ell_\infty$-error at most $1/50$, needs degree $\Omega(n)$.  The 
functional $\G_r$ defined by Grigoriev (see discussion above \eqnref{eq:B}) can be used to show an incomparable result: 
any sos polynomial of degree $(n-1)/2$ has error at least $\Omega(1/\log n)$ in approximating $f$ in $\ell_\infty$-norm.  

\begin{theorem}
Let $n$ be odd and $f: \{0,1\}^n \rightarrow \mathbb{R}$ be defined as $f(x) = (|x|-n/2)^2-1/4$.  Any sos polynomial of 
degree $(n-1)/2$ has error at least 
\[
\left(1- O(1/n) \right)\frac{\pi}{4} \frac{1}{\ln((n+1)/2)+\gamma+\ln(16)}
\]
in approximating $f$ in $\ell_\infty$ norm.  Here $\gamma \approx 0.577$ is the Euler-Mascheroni constant.
\end{theorem}

\begin{proof}
Let $h: \{0,1\}^n \rightarrow \mathbb{R}$ be a sos polynomial of degree $(n-1)/2$ approximating $f$ with 
$\ell_\infty$-error $\epsilon$.  Write $h(x) = f(x)+e(x)$ where $e$ is the function of ``errors'' satisfying 
$|e(x)| \le \epsilon$ for all $x \in \{0,1\}^n$.  Let $\delta_y: \{0,1\}^n \rightarrow \{0,1\}$ be the delta function on the boolean 
cube, where $\delta_y(x)=1$ if and only if $x=y$.  Recall that $\G_{n/2}(f) = \Sym^{uni}(f)(n/2)=(n/2-n/2)^2-1/4=-1/4$.  
By linearity of $\G_{n/2}$ we have
\begin{align*}
\G_{n/2}(f+e) = -1/4 + \G_{n/2}(e) &= -1/4 + \G_{n/2} \left(\sum_{y \in \{0,1\}^n} e(y) \delta_y \right) \\
& \le -1/4 + \epsilon \sum_{y \in \{0,1\}^n} | \G_{n/2}(\delta_y) | \enspace .
\end{align*}
On the other hand, $\G_{n/2}(f+e) \ge 0$ as $f+e$ is a sum-of-squares of polynomials of degree at most $(n-1)/2$ (property~4 in the proof of \thmref{thm:G}).  Thus 
\begin{equation}
\label{eq:eps_bound}
\epsilon \ge \left( 4 \sum_{y \in \{0,1\}^n} | \G_{n/2}(\delta_y) |\right)^{-1} \enspace.
\end{equation}
The main part of the proof will be to evaluate this sum.  

Let $L_i:\mathbb{R} \rightarrow \mathbb{R}$ be the degree-$n$ polynomial uniquely defined by
\[
L_i(z) = 
\begin{cases}
1 & z=i \\
0 & z \in \{0,1,2, \ldots, n\} \setminus \{i\}
\end{cases} \enspace .
\]
Then we see that $\G_{n/2}(\delta_y) = L_{|y|}(n/2)/ \binom{n}{|y|}$, and so 
\[
\sum_{y \in \{0,1\}^n} | \G_{n/2}(\delta_y)| = \sum_{k=0}^n |L_k(n/2)| \enspace .
\]
To do this sum, let us first simplify the summand
\begin{align*}
|L_k(n/2)| &= \frac{\prod_{a=0, a \ne k}^n |n/2-a|}{\prod_{a =0, a \ne k}^n |k-a|} \\
& = \frac{ \prod_{a=0}^n |n/2 -a|}{k! (n-k)! |n/2-k|} \\
& = \frac{1}{2^{n+1}}\frac{n!! n!!}{k! (n-k)! |n/2-k|} \\
& = \frac{n!}{2^{2n-1} \left(\frac{n-1}{2}\right)!^2} \binom{n}{k}\frac{1}{|n-2k|} \\
&= \frac{n}{2^{2n-1}} \binom{n-1}{(n-1)/2} \binom{n}{k} \frac{1}{|n-2k|} \enspace,
\end{align*}
where $n!!$ is defined as $\prod_{j=0}^{\ceil{n/2}-1}(n-2j)$, and we used $n!=n!!\cdot 2^{(n-1)/2}\cdot ((n-1)/2)!$
for odd $n$ in the penultimate equality.

For what comes next, it will be more convenient to express $|L_k(n/2)|$ in terms of $m=(n-1)/2$.  In this way, we 
obtain an expression defined for all $m$, rather than just odd $n$.  
\[
|L_k(m+1/2)|= \frac{2m+1}{2^{4m+1}} \binom{2m}{m} \binom{2m+1}{k}\frac{1}{|2m-2k+1|}
\]
Let $A(m)$ denote the sum over $k=0, \ldots, n=2m+1$, which is
\[
A(m)=  \frac{2m+1}{2^{4m+1}} \binom{2m}{m} \sum_{k=0}^{2m+1} \frac{1}{|2m-2k+1|} \binom{2m+1}{k} \enspace.
\]
By symmetry of the binomial coefficients we can multiply by $2$ and sum over only half of them, thereby 
removing the absolute values.
\[
A(m)=  \frac{2m+1}{4^{2m}} \binom{2m}{m} \sum_{\ell=0}^{m} \frac{1}{2\ell+1} \binom{2m+1}{\ell+m+1}
\]
Now we look at the difference between consecutive $A(m)$:

\begin{claim}
\[
A(m+1) - A(m) = \left( \frac{\binom{2(m+1)}{m+1}}{4^{m+1}} \right)^2
\]
\end{claim}

\begin{proof}
It is somewhat cumbersome to verify this claim directly.  We take the following approach.  Let $A(m)=B(m)C(m)$, where
\[
B(m)=\frac{2m+1}{4^{2m}} \binom{2m}{m}, \qquad C(m) = \sum_{k=0}^m \frac{1}{2k+1} \binom{2m+1}{k+m+1} \enspace.
\]
Note that 
\[
\frac{B(m+1)}{B(m)} = \frac{2m+3}{8(m+1)} \enspace.
\]
Since $B(0)=1$, this resolves to 
\[
B(m+1)=\frac{(2m+3)!!}{8^{m+1}(m+1)!}=\frac{2m+3}{4^{2(m+1)}}\binom{2(m+1)}{m+1} \enspace.
\]

By Zeilberger's algorithm \cite{PWZ97} we find a recurrence satisfied by the summand of $C(m)$.  
\[
\frac{2m+3}{2k+1} \binom{2m+3}{k+m+2} - \frac{8(m+1)}{2k+1}\binom{2m+1}{k+m+1} = \binom{2(m+1)}{m+k+1}-
\binom{2(m+1)}{m+k+2} \enspace .
\]
Summing this recurrence over $k=0,\ldots,m+1$ we find
\[
(2m+3) C(m+1) - 8(m+1) C(m) = \binom{2(m+1)}{m+1} \enspace .
\]
This means that 
\[
A(m+1) - \underbrace{\frac{8(m+1)}{2m+3}B(m+1)}_{B(m)}C(m) = \frac{B(m+1)}{2m+3}  \binom{2(m+1)}{m+1} \enspace,
\]
and in turn
\[
A(m+1) - A(m) = \frac{1}{4^{2(m+1)}} \binom{2(m+1)}{m+1}^2 \enspace.
\]
\end{proof}
As $A(0)=1$ this gives
\[
A(m) = \sum_{i=0}^m \left(\frac{\binom{2i}{i}}{4^i} \right)^2 \enspace .
\]

Luckily, the latter sum has already been asymptotically evaluated in the study of the quantum adversary
bound for the ordered search problem \cite{CL08}.  There it is shown that
\[
\sum_{i=0}^N \left(\frac{\binom{2i}{i}}{4^{i}} \right)^2 = \frac{1}{\pi}\left(\ln(N+1) + \gamma + \ln(16)\right) + O(1/N) 
\enspace ,
\]
where $\gamma \approx 0.577$ is the Euler-Mascheroni constant.

This gives
\begin{align*}
\sum_{y \in \{0,1\}^n} |\G_{n/2}(\delta_y)| =A((n-1)/2) &= \sum_{i=0}^{(n-1)/2} \left(\frac{\binom{2i}{i}}{4^{i}} \right)^2 \\
&=\frac{1}{\pi}\left(\ln((n+1)/2) + \gamma + \ln(16)\right) + O(1/n) \enspace .
\end{align*}
Plugging this into \eqnref{eq:eps_bound} gives the theorem.
\end{proof}
\end{document}